\documentclass[%
reprint,
superscriptaddress,
frontmatterverbose, 
preprintnumbers,
longbibliography,
 amsmath,amssymb,
 aps,
 pra,
notitlepage,
nofootinbib,
]{revtex4-1}

\usepackage{lipsum}

\usepackage{graphicx}% Include figure files
\usepackage{dcolumn}% Align table columns on decimal point
\usepackage{bm}% bold math
\usepackage{hyperref}
\usepackage{subfigure}
\hypersetup{
colorlinks=true,
linkcolor=blue,
filecolor=blue,
citecolor=blue,  
urlcolor=blue,
}

\newcommand{\mycomment}[1]{}

\usepackage{comment}

\usepackage{thmtools}
\usepackage{thm-restate}

\usepackage{enumerate} 
\usepackage{amssymb}
\usepackage{mathtools}
\usepackage{mathrsfs}
\usepackage{multirow}
\usepackage{bbm}

\usepackage{xcolor}

\definecolor{sanddune}{rgb}{0.59, 0.44, 0.09}
\definecolor{darkblue}{RGB}{0,0,102}
\definecolor{darkred}{rgb}{0.5,0.,0.}
\definecolor{BlueViolet}{RGB}{138,43,226}
\definecolor{SkyBlue}{RGB}{30,144,255}
\definecolor{DarkGreen}{RGB}{0,100,0}

\usepackage{amsthm}
\usepackage{amsmath}
\theoremstyle{plain}
\newtheorem{thm}{Theorem}
\newtheorem{lem}[thm]{Lemma}
\newtheorem{prop}[thm]{Proposition}
\newtheorem{cor}[thm]{Corollary}

\theoremstyle{definition}
\newtheorem{defn}{Definition}

\newcommand{\ket}[1]{|#1\rangle}
\newcommand{\bra}[1]{\langle #1|}
\newcommand{\bracket}[2]{\langle #1|#2\rangle}
\newcommand{\ketbra}[2]{|#1\rangle\langle #2|}

\newcommand{\cc}{\mathscr{C}}

\newcommand{\ba}{\begin{eqnarray}}
\newcommand{\ea}{\end{eqnarray}}

\DeclareMathOperator{\Tr}{Tr}

\newcommand{\tr}{\operatorname{Tr}}

\newcommand{\id}{\mathrm{id}}

\newcommand{\weicheng}[1]{ { \color{SkyBlue} \footnotesize (\textsf{wy}) \textsf{\textsl{#1}} }}

\usepackage{makecell}

\begin{document}
\newcommand{\onenorm}[1]{\left\| #1 \right\|_1}
\newcommand{\twonorm}[1]{\left\| #1 \right\|_2}
%\preprint{MIT-CTP/4902}

\newcommand{\ols}[1]{\mskip.5\thinmuskip\overline{\mskip-.5\thinmuskip {#1} \mskip-.5\thinmuskip}\mskip.5\thinmuskip} % overline short
\newcommand{\olsi}[1]{\,\overline{\!{#1}}} % overline short italic

\title{Complexity and order in approximate quantum error-correcting codes}

\author{Jinmin Yi}
\affiliation{Perimeter Institute for Theoretical Physics, Waterloo, Ontario N2L 2Y5, Canada}
\affiliation{Department of Physics and Astronomy, University of Waterloo, Waterloo, Ontario, Canada N2L 3G1}

\author{Weicheng Ye}
\affiliation{Perimeter Institute for Theoretical Physics, Waterloo, Ontario N2L 2Y5, Canada}

\author{Daniel Gottesman}
\affiliation{Joint Center for Quantum Information and Computer Science (QuICS) and Computer Science
Department, University of Maryland, College Park, MD 20742, USA}

\author{Zi-Wen Liu}
\affiliation{Yau Mathematical Sciences Center, Tsinghua University, Beijing 100084, China}
%\affiliation{Perimeter Institute for Theoretical Physics, Waterloo, Ontario N2L 2Y5, Canada}

\date{\today}% It is always \today, today,

\begin{abstract}
We establish rigorous connections between quantum circuit complexity and approximate quantum error correction (AQEC) capability, two properties of fundamental importance to the physics and practical use of quantum many-body systems, covering systems with both all-to-all connectivity and geometric scenarios like lattice systems in finite spatial dimensions. To this end, we introduce a type of code parameter that we call subsystem variance, which is closely related to the optimal AQEC precision.
Our key finding is that, for a code encoding $k$ logical qubits in $n$ physical qubits, if the subsystem variance is  below an $O(k/n)$ threshold, then any state in the code subspace must obey certain circuit complexity lower bounds, which identify nontrivial ``phases'' of codes. Based on our results, we propose $O(k/n)$ as a boundary between subspaces that should and should not count as  AQEC codes. This theory of AQEC provides a versatile framework for understanding quantum complexity and order in many-body quantum systems, generating new insights for wide-ranging physical scenarios, in particular topological order and critical quantum systems which are of outstanding importance in many-body and high energy physics. 
We observe from various different perspectives  that roughly  $O(1/n)$ represents a common, physically significant ``scaling threshold'' of subsystem variance for  features associated with nontrivial quantum order.

\end{abstract}

\maketitle

%\tableofcontents

\section{Introduction}

A pillar of quantum information science and technology, quantum error correction (QEC) has been extensively studied as a means to protect quantum information from noise and errors for the purpose of realizing the potential advantages of quantum computation in practice \cite{shor95,gottesman1997,nielsen2000quantum}. Remarkably, in recent years, it has  become increasingly evident that the concept of QEC carries broad importance  in fundamental physics,  extending far  beyond its original realm. 
Most notably, QEC plays fundamental roles in our understanding of topological order~\cite{KITAEV20032} and  anti-de Sitter (AdS)/conformal field theory (CFT) correspondence~\cite{Almheiri2015} that stand at the frontier of many-body physics and quantum gravity, respectively.

The idea behind the standard notion of QEC is to encode the logical system into a suitable code subspace in such a way that the logical information is effectively ``hidden'' by entanglement and thus remains recoverable under noise. Owing much to the clean yet powerful scheme of stabilizer codes~\cite{gottesman1997}, it is customary in the study of QEC to seek and understand quantum codes that enable exact recovery.  
However, generalized notions of codes that may only achieve QEC in an approximate manner could be adequate for practical purposes and outperform exact QEC codes in various ways~\cite{Leung97,CGS05:aqec}. Furthermore, they
encompass a much broader range of scenarios innate especially in physical contexts, underscoring the fundamental importance of approximate quantum error correction (AQEC) from both practical and theoretical perspectives.    

Despite the extensive study of QEC, our knowledge on AQEC codes is limited to various scattered situations (e.g.,~ examples in spin chains \cite{BCSB19:aqec}, covariant codes \cite{Hayden_2021,faist20,Woods2020continuousgroupsof,PhysRevLett.126.150503,PhysRevResearch.4.023107,Zhou2021newperspectives,liu2022approximate,KongLiu22}, quasi-exact codes~\cite{quasi,Wang2022}), with a fundamental understanding of such codes remaining elusive. It is worth mentioning a striking finding that an extremely small imprecision tolerance suffices to enable the decoding radius to match the code distance~\cite{CGS05:aqec, BGG22:singleton-aqec}, in stark contrast with the situation of exact QEC, signifying the intrinsic distinction in nature.  In the literature, the notion of AQEC commonly just means that the imprecision is vanishingly small in system size.  However, this can even be naturally achieved by a trivial encoding defined by appending a series of garbage states to the logical state for random local noise, simply because the chance of logical information being affected is vanishingly small (concrete analysis in App.~\ref{app:redundant}).    
This suggests that our current understanding of AQEC is too coarse.

In addressing this predicament, we establish a general theory of AQEC codes based on  quantum circuit complexity whose importance permeates quantum computation, complexity theory, and physics (see e.g.,~Refs.~\cite{aaronson2016complexity,Freedman2014,wen2013topological,susskind2016complexity}), encompassing scenarios both with and without geometric locality, encompassing scenarios both with and without geometric locality.
More specifically, we define a code parameter called \emph{subsystem variance} that characterizes the fluctuation of marginals of the physical system and is closely related to the usual notions of AQEC imprecision. We derive critical values of the subsystem variance that scale roughly as $O(k/n)$, below which the entire code subspace is subject to nontrivial circuit complexity lower bounds depending on the geometry. The conditions are nearly optimal in certain regimes and provide meaningful criteria for interesting codes in general, as supported by concrete examples.   
From a code perspective, our results suggest that it is reasonable to consider $O(k/n)$ a  boundary  between subspaces that should be regarded ``acceptable'' AQEC codes and those that should not be.  
Our theory offers not only a fundamental understanding of nontrivial AQEC codes, but also useful methods for the widely important but notoriously difficult problem of  proving circuit lower bounds.  
The wide applicability of our theory is demonstrated by  various examples arising from both quantum computation and physics.

Remarkably, our AQEC framework and results have broad applications in many-body physics, bridging general information-theoretic properties with quantum physical features in many ways. In particular, we gain new insights into nontrivial quantum order or long-range entanglement that underpin ``exotic'' quantum features including topological order and criticality that are of central importance in condensed matter and high energy physiscs' contexts.   
 For topological order, we find that AQEC offers a unifying framework for rigorously understanding the relationship between strict notions of gapped topological order and the long-range entanglement and topological entanglement entropy (TEE) signatures.  For critical quantum systems, we show that a power-law AQEC imprecision is a fundamental nature of CFT codes that emerge at low energies and discuss how our theory may provide insights into quantum gravity through AdS/CFT.  
 It is notable that a roughly $O(1/n)$ imprecision scaling naturally arises  as some kind of ``threshold'' in several different situations. 
 
This paper is organized as follows. In Sec.~\ref{sec:prelim} we provide formal definitions of the AQEC parameters and circuit complexities that we will use. In Sec.~\ref{sec:complexity} we introduce our primary results about circuit lower bounds from AQEC error and provide some intuitions. In Secs.~\ref{sec:to} and \ref{sec:CFT}, we discuss AQEC and the complexity results in the contexts of topological order and critical quantum systems, respectively, demonstrate the physical significance of our study. 
We conclude with some remarks and outlook in Sec.~\ref{sec:discussion}.
In the appendices we provide technical details as well as extended results and discussions.

\section{Preliminaries\label{sec:prelim}}

Here we will work with multi-qubit systems, namely, those living in a Hilbert space given by the tensor product of multiple qubit Hilbert spaces.\footnote{The qubit assumption is not essential;  generalizations to higher local dimensions are straightforward.}  
The notions of locality and geometry associated with  such a multi-qubit system are captured by an {adjacency graph}, 
the edges of which define the connection relations among the nodes (qubits).
Two prototypical types of adjacency graph are complete (all-to-all) graphs and local lattices embedded in finite spatial dimensions, with the latter incorporating geometric locality that is essential in physical contexts.  

We first lay the groundwork for  our study of AQEC.
We call any  $2^k$-dimensional subspace of an $n$-qubit Hilbert space  an $(\!(n,k)\!)$ quantum code  ($k<n$) as it represents an encoding of a $k$-qubit logical system into an $n$-qubit physical system, and any pure state within this subspace is called a code state. Of course, it may not be a good QEC code  as QEC requires intricate structures.  
A theme of this work is to understand the meaning of the deviation from ideal QEC codes, which is generic and serves as the basis for the theory for AQEC.
This deviation can naturally be quantified by how well the recovery can restore the logical information after the system undergoes noise; specifically, for encoding channel $\mathcal{E}$
 and some noise channel 
$\mathcal{N}$, the \emph{QEC inaccuracy} is defined as
\begin{equation}
\tilde\varepsilon\left(\mathcal{N}, \mathcal{E}\right)\coloneqq\min _{\mathcal{R}} P\left(\mathcal{R} \circ \mathcal{N} \circ \mathcal{E}, \id_L\right),
\end{equation}
i.e.,~the minimum distance (here we adopt the channel purified distance $P$; see App.~\ref{app:aqec} for the detailed definition) 
between the overall logical channel after recovery $\mathcal{R}$  and  the logical identity $\id_L$.
%\begin{defn}
From what features of the code does such QEC inaccuracy originate? To understand this, we 
introduce a type of parameter intrinsically associated with the code space $\mathfrak{C}$ (image of $\mathcal{E}$) that we call \emph{subsystem variance}: 
%\end{defn}
\begin{equation}
    \varepsilon_\mathsf{G}(\mathfrak{C},d) \coloneqq \max_{\psi\in\mathfrak{C},|S|\leq d}\onenorm{\psi_S - \Gamma_S},
\end{equation}
where $\Gamma\coloneqq\frac{1}{2^k} \sum_{i=1}^{2^k}   |\psi_i\rangle\langle\psi_i|$  is the statistical average of $\{\psi_i\}$ that span $\mathfrak{C}$, that is,~the maximally mixed state of $\mathfrak{C}$, 
and $S$ is a connected (local) subsystem with respect to the adjacency graph $\mathsf{G}$ (subscript $S$ denotes the reduced state on $S$).  
Here $d$ should be treated as a tunable variable that generalizes the notion of code distance. 
Intuitively, the subsystem variance limits the accessible information from the subsystems and is thus closely tied to entanglement and QEC properties. In particular, it bounds the  violation of the Knill--Laflamme QEC conditions \cite{PhysRevA.55.900} and broadly characterizes the QEC inaccuracy. As an extreme instance, under the same locality restriction, 
\begin{equation}
\label{eq:error_variance}
 \frac{\varepsilon}{4} \leq {\tilde\varepsilon} \leq 2^{k/2}\sqrt{\varepsilon} 
\end{equation}
 for any noise represented by replacement channels, including, for example, erasure,complete depolarizing and reset channels (some relations between QEC error and the violation of Knill--Laflamme conditions are known~\cite{BO10:AKL,Ng10}). 
A complete form of this result and detailed proof and discussion can be found in App.~\ref{app:aqec}.  
In addition, we can also establish two-way bounds that relate subsystem variance with coherent information, a well-known quantity that characterizes quantum information loss~\cite{cohinfo,PhysRevA.55.1613,CoherinfoAQEC2001} (details given in App.~\ref{sec:coh-info}).
The physics discussions in later sections mainly concern scenarios where  $\varepsilon$ and $\tilde\varepsilon$ convey similar messages. They may generally be referred to as code/AQEC error at appropriate instances.

Next, we formally define the notions of quantum circuit complexity that will be studied. 
Generally, for an $n$-qubit quantum state $\ket{\psi}$, the (quantum circuit) complexity associated with the adjacency graph $\mathsf{G}$, denoted by $\cc_\mathsf{G}(\psi)$, is defined as the minimum depth (number of layers) of 2-local (with respect to $\mathsf{G}$) quantum circuits that generate $\ket{\psi}$ from $\ket{0}^{\otimes n}$. More precisely,
%The \emph{all-to-all (quantum circuit) complexity} (associated with complete graphs) of an $n$-qubit quantum state $\ket{\psi}$ is defined as the minimum depth (number of layers) of  $2$-local quantum circuits that generate $\ket{\psi}$ from $\ket{0}^{\otimes n}$. More precisely,
\begin{equation}
    \cc_\mathsf{G}(\psi) \coloneqq \min\left\{l: \ket{\psi} = \prod_{i=1}^l U_i\ket{0}^{\otimes n}\right\},
\end{equation}
where $U_i$'s must be a tensor product of disjoint 2-qubit unitary gates acting on the connected nodes according to $\mathsf{G}$.
The two standard scenarios specifically discussed here are all-to-all quantum circuit complexity associated with complete graph $\mathsf{g}$, and geometric quantum circuit complexity associated with some adjacency graph that encodes certain underlying geometry. In the main text we specifically discuss $D$-dimensional integer lattices  ${\mathsf{G}}_D$, with generalizations to arbitrary graphs being feasible.

%Then for the physical scenarios with geometry encoded in some adjacency graph $\mathsf{G}$, the \emph{geometric (quantum circuit) complexity} of $\ket{\psi}$, denoted as $\cc_{\mathsf{G}}(\psi)$, is defined analogously, with the difference being that only $2$-qubit gates acting on nearest neighbors in $\mathsf{G}$ are allowed. %endowed with a Euclidean metric.
%Here we specifically consider $D$-dimensional integer lattices  ${\mathsf{G}}_D$, with generalizations to other graphs being feasible.
The \emph{$\delta$-robust} versions  of these complexities, denoted by $\cc^\delta$,  are  defined as the minimum corresponding complexity of any state $\ket{\psi'}$ within the $\delta$-vicinity of $\ket{\psi}$ in trace norm, namely,
\begin{equation}
    \onenorm{\psi'-\psi}\leq\delta
\end{equation}

It should be noted that our results and discussions are robust under various modifications of the above complexity definitions. 
First, generalizing 2-locality to  $t$-locality for any finite $t$ only introduces constant factors to the results.\footnote{More explicitly, it can be seen from the proofs that if we consider $t$-local gates instead of $2$-local gates, then our circuit complexity bounds for the all-to-all and geometric cases hold with an extra $1/\log t$ factor and an extra $1/(t-1)$ factor, respectively.}  
Furthermore,  several  physical variants of the setting, including quasi-local gates with fast decaying tails, quasi-adiabatic evolutions \cite{HastingsWen}, and more general lattices, are expected to retain the relevant messages in this work.

While our results apply to any  specific $n$, one is often mostly interested in the asymptotic scalings in the thermodynamic (large $n$) limit. 
In addition to the standard  Bachmann--Landau notation using $O,o,\Omega,\omega, \Theta$ symbols  (see e.g., Refs.\ \cite{knuth1976bigO,wiki-big-o-notation} for an extensive introduction), we shall also use  the  ``soft'' notation with a tilde on top, which hide  polylogarithmic factors that are insignificant in our context; explicitly, for $A\in\{O,\Omega,\Theta\}$, $\tilde A(f(n))$ means $A(f(n)\,{\rm polylog}(n))$ for some polylog function, and for $a\in\{o,\omega\}$, $\tilde{a}(f(n))$ means $a(f(n)\,{\rm polylog}(n))$ for any polylog function.

\section{Circuit complexity from approximate quantum error correction}
\label{sec:complexity}

We now introduce our key results on quantum circuit complexity from AQEC for both all-to-all and geometric cases. The log symbols denote the logarithm to base 2.       $H_2(p) = -p\log p -(1-p)\log(1-p)$ is the binary entropy function; whenever it appears it is assumed that $p<1/2$.  
\begin{thm}
\label{thm:all_comp}
Given an $(\!(n,k)\!)$ code $\mathfrak{C}$,
the $\delta$-robust all-to-all quantum circuit complexity of any code state $\ket{\psi}\in\mathfrak{C}$ satisfies $\cc^{\delta}(\psi)>\log d$, if $H_2(\varepsilon_\mathsf{g}(\mathfrak{C},d)/2 + \delta/2) < k/n$ with $\mathsf{g}$ being the complete graph ($\varepsilon$ is defined with respect to any $d$ qubits) where $\varepsilon+\delta < 1$.  
\end{thm}
\begin{thm}
\label{thm:geo_comp}
Given an $(\!(n,k)\!)$ code $\mathfrak{C}$, the $\delta$-robust geometric circuit complexity with respect to adjacency graph ${\mathsf{G}}_D$ embedded in a $D$-dimensional integer lattice\footnote{The topology of the base manifold does not affect the result.} of any code state $\ket{\psi}\in\mathfrak{C}$  satisfies $\cc_{{\mathsf{G}}_D}^{\delta}(\psi)>(d^{1/D}-1)/2$, if $H_2(\varepsilon_{\mathsf{G}_D}(\mathfrak{C},d)/2 + \delta/2) < k/n$  where $\varepsilon+\delta < 1$.
\end{thm}
\noindent {\bf Remark.} The results are applicable to different notions of AQEC error. In particular, using (\ref{eq:error_variance}), the error conditions can be alternatively expressed  in terms of QEC inaccuracy $\tilde\varepsilon$, substituting $\varepsilon_{\mathsf{G}}(\mathfrak{C},d)$ with $4\tilde\varepsilon(\mathcal{N},\mathcal{E})$ for replacement channels $\mathcal{N}$ acting on any $d$-local subsystems with respect to the associated adjacency graph $\mathsf{G}$.  
Then, as mentioned,
our approach can be generalized to obtain circuit complexity bounds for arbitrary adjacency graphs  (see App.~\ref{app:proof}). Additionally, it is worth noting that our complexity results indicate intrinsic circuit complexities of code states themselves, which should not be confused with the depth of the state-independent encoding circuits that are blind to the input logical states.

%\smallskip
Evidently, our results cover exact QEC codes as special cases (e.g.,~recovering the stabilizer code result of Ref.~\cite{AnshuNirkhe20} in the NLTS context and the long-range entangled property of the toric code), which can be seen by noting that $\varepsilon_\mathsf{\mathsf{G}}(\mathfrak{C},d-1) =0$ for any $[[n,k,d]]$ code and any adjacency graph $\mathsf{G}$ so that our code error conditions are automatically satisfied.
It is worth emphasizing that the above results  reflect universal complexity features of the entire code spaces that encompass arbitrary superpositions of special wavefunctions in the code, which are important but rarely understood in physics contexts. From the perspective of proving circuit complexity lower bounds,  our approach can be used to establish bounds for specific states beyond the applicability of the theorems (see App.~\ref{app:proof}). A physically interesting family of examples is given by what we call momentum codes, which are discussed in detail in App.~\ref{app:momentum}.
Another noteworthy point is that AQEC properties are able to guarantee  the intrinsic
all-to-all circuit complexity, which  is not constrained by  geometry and spatial locality.
Finally, note that the results conversely indicate lower bounds on code error that depend on the lowest complexity of code states (see App.~\ref{app:proof} for formal statements);  it could be interesting to consider what code error is needed such that the entire code is subject to circuit complexity bounds.

With improved forms of the results and  detailed proofs provided in App.~\ref{app:proof}, we now distil the core intuitions that apply generally to any connectivity, covering all-to-all and geometric cases.
Our results roughly say the ``distance'' (noise size) under which a code can maintain a sufficiently small code error indicates circuit complexity lower bounds.  An overall conceptual message is that  higher complexity is generally associated with smaller code error and larger code rates.
The main proof idea, adapting a method in Ref.~\cite{AnshuNirkhe20} to our AQEC setting, goes as follows.
Suppose a code state is generated by a circuit $Q$ of some low depth from $\ket{0}^{\otimes n}$ where each qubit's effects are confined within its light cone determined by $Q$.
Now run this circuit backwards (apply $Q^\dagger$) on  the maximally mixed code state $\Gamma$.  Using the light cone properties, one finds that if  $\varepsilon$ within the light cone scale of $Q$ ($O(\exp({n}))$ for all-to-all circuits and $O(n^D)$ for $D$-dimensional circuits) is small, then the output of the backward circuit  is  well approximated by $\ket{0} $ locally and therefore the entire system has small entropy due to subadditivity.  Making $\varepsilon$ sufficiently small  leads to contradictions with the entropy of $\Gamma$ directly determined by $k$ which, in turn, implies complexity lower bounds because the light cones are too small for consistency.

%Although our results apply 
The distinction of whether the circuit complexity of a system is $O(1)$ (i.e.,~finite in the large $n$ limit) holds exceptional importance in both physics and complexity theory.   States with $\omega(1)$ (superconstant)  complexity, often referred to as long-range entangled states when a proper notion of geometric locality is present, are generally associated with certain kinds of nontrivial quantum order %(see Secs.~\ref{sec:to} and \ref{sec:CFT}) 
and play central roles in the theories of phases of matter and Hamiltonian complexity etc.\footnote{In both contexts, it is sometimes desirable to consider $\omega(\mathrm{loglog}(n))$ complexity; however, 
this difference is inconsequential in our results.}  The key implications of our theory are summarized in the following corollary.
\begin{cor}\label{cor:scaling}
Given an $(\!(n,k)\!)$ code $\mathfrak{C}$ with subsystem variance $\varepsilon_\mathsf{G}(\mathfrak{C},d)$ where $d = \omega(1)$. Suppose $H_2(\varepsilon_\mathsf{G}(\mathfrak{C},d)/2)< k/n$, which is satisfied in particular when 
\begin{itemize}
    \item $k = \Tilde{O}(1), \varepsilon = \Tilde{o}(1/n)$;
    \item  $k = \Omega(n), \varepsilon =o(1)$.
\end{itemize}
Then for any code state $\ket{\psi}\in\mathfrak{C}$, it holds that  $\cc_{\mathsf{G}}(\psi) = \omega(1)$.

\end{cor}

Fig.~\ref{fig:phase} depicts schematic circuit complexity ``phase diagrams'' for any $d$ and $\mathsf{G}$ in terms of the corresponding $\varepsilon$ as well as $\tilde\varepsilon$ (for replacement channels) over $k$. 
\begin{figure}
\includegraphics[width=\columnwidth]{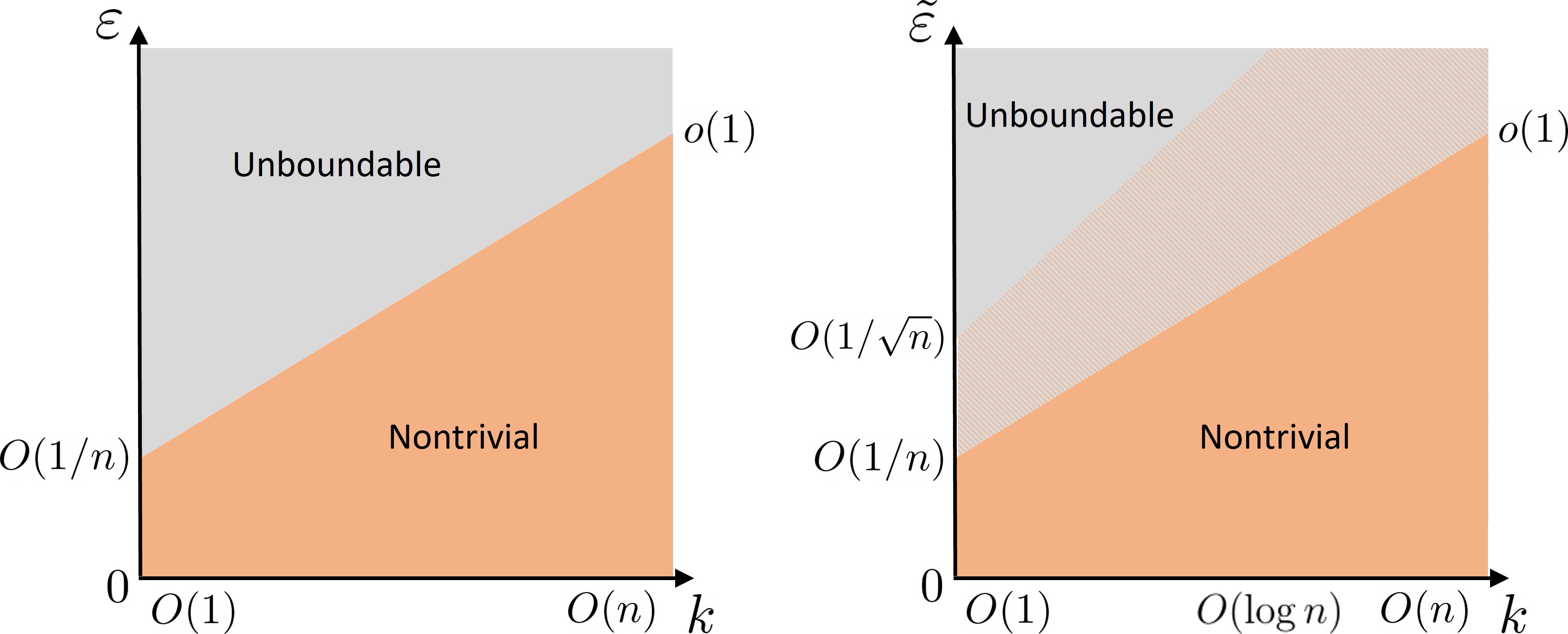}
\caption{\label{fig:phase}Schematic circuit complexity ``phase diagrams'' of general quantum codes for any $d$ and $\mathsf{G}$. %The scalings are sketchy.  
``Nontrivial'' and ``unboundable'' mean that our complexity bounds hold for any code state, or are inapplicable, respectively; The intermediate regime for $\tilde\varepsilon$ can be regarded an expanded boundary where the applicability of the complexity bounds depends on the code.}
\end{figure}
Most importantly, our results identify ``nontrivial'' regimes of code parameters---any state that belongs to a code within these nontrivial regimes is subject to our circuit complexity lower bounds.

 Note that the critical scalings below which our bounds can apply can be roughly achieved in a naive manner (see the redundant encoding in App.~\ref{app:redundant}), suggesting that our code error conditions for circuit lower bounds represent meaningful conditions for AQEC codes. Besides, the conditions are nearly tight  for small $k$ as evident from the Heisenberg chain code (see App.~\ref{app:heisenberg} for details).

The universality of our framework enables  applications in  an exceptionally broad range of scenarios in coding theory and physics. 
In Table~\ref{tab:examples}, we summarize the properties of various representative types of AQEC codes originating from diverse contexts, which also provide meaningful examples for different parameter and complexity regimes. We will study the CFT codes and momentum codes in more detail in this work.  
\begin{table}\footnotesize
\begin{ruledtabular}
    \centering
    \begin{tabular}{lccc}
    Code&$k$&$\varepsilon$\footnote{with respect to suitable superconstant $d$  }&\makecell{Code space\\ complexity}

    %{Complexity}
    \\
    \hline

    ETH energy window & $\Omega(n)$& $O(\exp(-n))$&Nontrivial$^\ast$\\
    Random unitary code\footnote{with $1-O(\exp(-n))$ probability} & $\Omega(n)$& $O(\exp(-n))$&Nontrivial$^\ast$\\
    Topological code & $O(1)$& $O(\exp(-n))$&Nontrivial\\
      ``Good'' AQLDPC code & $\tilde\Omega(n)$& $O(1/\mathrm{polylog}(n))$&Nontrivial$^\ast$\\
      Heisenberg chain code & $\Omega(\log n)$&$\tilde{\omega}(1/n)$&Unboundable\\
    \textit{CFT low energy sector} & $O(1)$ & $\Theta(1/n^{\varDelta/D})$\footnote{$\varDelta$ is the minimum scaling dimension of the CFT; see Sec.~\ref{sec:CFT}.}&Indefinite$^\ast$\\
{\textit{Momentum code}}&$O(\log n)$&$O(1)$& Unboundable\\

    \end{tabular}
    \end{ruledtabular}
    \caption{Various representative examples of AQEC codes. The ETH and Heisenberg chain codes (refined analysis in App.~\ref{app:heisenberg}) are defined in Ref.~\cite{BCSB19:aqec}; the good approximate quantum LDPC code refers specifically to the spacetime Hamiltonian construction in Ref.~\cite{Bohdanowicz_2019};
    the italicized ones are explicitly studied in this paper.    
    The rightmost column lists the complexity results deduced from our theory, where the asterisks specify applicability to all-to-all complexity and the remaining entries concern geometric complexity in their respective native dimensions.
    }
    \label{tab:examples}
\end{table}

\mycomment{
In particular, the  momentum codes are a new type of AQEC codes constructed from spin systems with nontrivial lattice momentum, which have physical significance~\cite{GioiaWang} and exhibit intriguing properties. For the simplest case, consider the family of 1D states
\begin{equation}
    |W_p\rangle=\frac{1}{\sqrt{n}}\sum_{x=1}^n e^{ipx}|\tilde x\rangle,~ p=\frac{2\pi m}{n},~ m=0,1,\ldots,n-1,
\end{equation}
where $\ket{\tilde{x}}$ denotes the all-0 state with a quasiparticle excitation 1 at the $x$-th site. One can think of $\ket{W_p}$ as a generalized $W$-state with lattice momentum $p$, i.e.,~$\hat{T}|W_p\rangle=e^{ip}|W_p\rangle$ under lattice translation operator $\hat{T}$.  We call spaces spanned by such states momentum codes, where the momentum is the logical information being encoded (so the code parameters can characterize the local detectability of momentum, which is of potential physical interest).  
The momentum codes naturally reside near the $O(1/n)$ nontriviality boundary and demonstrate a remarkable transition across the boundary upon code ``fragmentation'':
The full code consisting of all $\ket{W_p}$ falls right outside the nontrivial regime, but  smaller fragments of the code with close momenta can exhibit smaller code error and thus transition into the nontrivial regimes, enabling us to show that quasiparticle states with any momentum has superconstant circuit complexity, in line with Ref.~\cite{GioiaWang}.     A detailed exposition of the momentum codes  can be found in App.~\ref{app:momentum}.

}

\section{Topological order\label{sec:to}}
Our complexity results  shed a new light on topological order, a widely studied concept in modern condensed matter physics that characterizes exotic quantum phases of matter arising from many-body entanglement.  A central problem in the study of topological order is to identify simple criteria or indicators for states associated with systems  with topological order. 
As signified by the prototypical example of the toric code \cite{KITAEV20032}, QEC is a representative feature (and application) of topologically ordered systems.  Indeed, QEC properties underlie the well-established topological quantum order (TQO) condition, which is tied to strong physical notions of topological order such as gap stability \cite{BHV06,BHM10:TQO}, essentially demanding that the state belongs to an almost exact QEC code 
with macroscopic (at least $\mathrm{poly}(n)$) distance.  
On the other hand, based directly on many-body entanglement properties, there are two other prominent characteristics,  oftentimes considered definitions, for states with  topological order:  long-range entanglement \cite{ChenGuWen10:lu,wen2013topological}, and topological entanglement entropy (TEE) \cite{kitaev2006topological,levin2006detecting}.  
Despite  extensive study and usage of all three conditions, 
their relationship has not been systematically understood. We now demonstrate that AQEC provides a general framework that allows us to rigorously compare the TQO (code) and entanglement conditions, thereby sharpening our understanding of topological order.  

First, a direct implication of our results is a general quantitative understanding of the gap between TQO and long-range entanglement conditions, with the former being strictly stronger than the latter.   To be more specific, recall that long-range entanglement  means superconstant circuit complexity (in the current context one considers finite $D$). 
According to Corollary~\ref{cor:scaling}, the code property requirements in TQO can be  relaxed from exponentially small $\varepsilon$ under macroscopic distance to  $\varepsilon=\Tilde{o}(1/n)$ under any superconstant (e.g., logarithmic) distance, while still ensuring that all code states are long-range entangled.  This is a substantial relaxation that is expected to encompass wide-ranging physical situations.  
However, in the literature, the notions of long-range entanglement and topological order are frequently lumped together, especially for gapped systems. Our observation elucidates the discrepancy between the two notions in terms of AQEC parameters, sharpening the characterization of topological order. %as QEC is inherent to the rigorous physical notion of topological order given its connection to gap stability~\cite{BHM10:TQO}. 

Next, let us consider TEE, which widely serves as a simple information-theoretic signature for topological order.   More concretely, consider the standard 2D setting, where ground states of  gapped systems commonly obey an area law, with possible subleading corrections originated from long-range entanglement inherent in topologically ordered systems, based on which 
TEE is defined. Specifically, suppose the entanglement entropy of any contractible subsystem $A$ has the form 
\begin{equation}\label{eq:area}
S(A)=al(A)-\gamma + o(1),    
\end{equation}
where the first term manifests the area law, with $a$ being some constant and $l(A)$ being the length of the boundary of $A$, and the correction $\gamma$ is the  TEE which is expected to be a universal constant signifying topological order. 
For general AQEC code states, we show the following: 
%\begin{prop}
  % Consider an $(\!(n,k)\!)$ code defined on a 2D lattice on a torus.  Suppose that  $\varepsilon=o(1/n)$ for any contractible region of size $d$ and there exists a code state with area-law entanglement. Then,  in the thermodynamic limit, the TEE of any code state satisfies
   % \begin{equation}
       % \gamma\geq k/\max\{2,2\lfloor n/2d\rfloor\}.
    %\end{equation}
%Specifically, we have the best bound $\gamma\geq k/2$, which is saturated by abelian topological order, if  the code conditions are satisfied by some $d>n/4$.

%As a corollary, a 2D area-law state with zero TEE does not belong to any code that achieves $\varepsilon=o(1/n)$ on linear-size contractible regions.
%\end{prop}

\begin{prop}
    (Informal) 
  Consider an $(\!(n,k)\!)$ code with an area-law code state defined on a 2D torus.  Suppose  $\varepsilon=o(1/n)$ for any contractible region of linear size; then all code states have nontrivial TEE.
   
   Specifically, the abelian topological order saturates the best TEE lower bound of $k/2$ from our approach.

As a corollary, a 2D area-law state with trivial TEE does not belong to any code that achieves $\varepsilon=o(1/n)$ on linear-size contractible regions.
\end{prop}

Full formal results are given in App.~\ref{app:tee}.
The corollary provides a simple physical diagnostic for ``bad'' code states. 
Also note that the best lower bound $\gamma\geq k/2$ is attained when $d>n/4$, and this $n/4$ can be improved to $cn$ with any $c>0$ if we further require a local recoverability feature (which is expected to hold generally for topological order) by leveraging the Expansion Lemma from Ref.~\cite{Flammia2017limitsstorageof}.

This result  is proven using the prescription from Ref.~\cite{IssacKim-TOEntropy}. The main idea is to apply the Markov entropy decomposition~\cite{Hastings2011MED} to relate $k$ and a signed sum of subregion entropies in which all the area law terms are canceled out, leaving only TEE with corrections due to subsystem variance. When $\varepsilon=o(1/n)$, the corrections turn out to be vanishingly small, ensuring nontrivial bounds on TEE.  
Conversely, when $\varepsilon=\omega(1/n)$, there is no nontrivial bound because TEE cannot overshadow the corrections.
We further note that this result does not hinge on a strict area law, i.e.,~small fluctuations of the correction $\gamma$ are allowed, in which case the lower bound is for the average TEE.
By considering a deformation of the toric code through adding string tension (such that the code states are string-net wavefunctions with tension)~\cite{Wen:ToricCodeTension}, we can construct a physically interesting example of a code family  with tunable AQEC error where the TEE vanishes as the AQEC error increases, in accordance with our results (also see App.~\ref{app:tee}).  Details and extended discussion of this section can be found in App.~\ref{app:tee}.

To conclude, a key takeaway is that TEE and long-range entanglement have a similar $\sim 1/n$ robustness against AQEC error.
It is worth emphasizing that our discussion applies to any subspace, not hinging on gapped ground spaces (associated with Hamiltonians) as conventional in the context of topological order.

\section{Critical systems and conformal field theories\label{sec:CFT}}

Critical quantum systems, widely described by conformal field theory (CFT),  represent another prominent type of quantum order with wide-ranging physical significance.
The nature of critical systems, specifically their gaplessness and scale invariance, suggests a universal presence of highly nontrivial  entanglement that supports interesting quantum coding properties. 
Notably, in the context of quantum gravity, the concept of CFT codes is expected to play a pivotal role, in light of the fundamental connection \cite{Almheiri2015,harlow2018tasi}  between QEC and AdS/CFT correspondence \cite{Maldacena1999,Witten1998}.    

Physically, it is most natural to consider the low-energy sectors as code spaces, where $k$ does not scale with the system size and all states are CFT states. As we will now demonstrate, they generally give rise to intrinsic AQEC codes whose properties are closely connected to the physics of the system and can be concretely analyzed by employing techniques from the field of CFT. 
Let the system be defined on a hypersphere $S^D$ of $D\in\mathbb{Z}^+$ spatial dimensions.  Using the state-operator correspondence~\cite{TASI_CFT2016} on the cylinder geometry $S^D\times R$, it can be shown that the one-point functions $\langle \phi_\beta|\phi_\alpha|\phi_\gamma\rangle$ for code states $|\phi_{\beta,\gamma}\rangle$ and  local ($O(1)$-size) primary operator $\phi_\alpha$, which can be  related to the Knill--Laflamme conditions for QEC, exhibit the scaling behavior~\cite{cardy1984conformal,cardy1988universal}
\begin{equation}
\label{eq:1pf-scaling}
\langle \phi_\beta|\phi_\alpha|\phi_\gamma\rangle =\Theta\left( \frac{1}{ n^{\Delta_\alpha/D}}\right),
\end{equation}
where $n$ is the total system size and $\Delta_\alpha$ is the scaling dimension of $\phi_\alpha$.  
These one-point functions are studied for extracting the conformal data in the CFT literature  \cite{cardy1986,blote1986,zou2020,Hu2023}. 
Here we can use them to compute the AQEC error scaling.  For the sake of nontrivial circuit complexity arguments, consider  operators  (that are not necessarily spatially local) of size $O(\mathrm{loglog}(n))$, which are sufficiently small and can thus be well approximated by a product of local operators via proper renormalization group flow so that the scaling in (\ref{eq:1pf-scaling}) is expected to hold up to insignificant (at most $\mathrm{polylog}(n)$) factors. Then it can be shown that  $\varepsilon$ follows a power law scaling with the exponent determined by the minimum scaling dimension $\varDelta$, precisely, 
\begin{equation}
    \varepsilon=\Tilde{\Theta}(n^{-\varDelta/D}).
\end{equation}
   Crucially, this shows that the polynomial AQEC error is a  fundamental nature of CFT, different from e.g.~topological order. 
App.~\ref{app:cft} provides a more detailed exposition of the derivation.

As per Corollary~\ref{cor:scaling}, our theory has the following implication that could be of particular interest in relation to holography: if the minimum scaling dimension $\varDelta>D$ (so $\varepsilon = \Tilde{o}(1/n)$), then any state in the CFT code is not just long-range entangled (as can  also be inferred from general spatial correlation properties of CFT using Lieb--Robinson-type arguments \cite{Lieb1972,BHV06}), but actually has the fundamentally stronger feature of being intrinsically nontrivial, i.e.,~has superconstant all-to-all complexity. 
This all-to-all property is morally congruent with the ultrastrong long-range interactions and  chaotic behaviors (think e.g.~SYK model) which are strong physical signatures of nontrivial gravity duals in  AdS/CFT \cite{kitaev2015,PhysRevD.94.106002,Heemskerk_2009,maldacena2016,witten2007threedimensional}, providing new insights into the duality from a quantum information perspective. 
Indeed, large scaling dimension is associated with large central charge \cite{Hellerman:2009bu, Rattazzi:2010gj}  and bulk field mass \cite{GUBSER1998105,Witten1998}, which are in accordance with strong coupling and complexity.  While our understanding of their relationship is incomplete, 
the connections among all these perspectives indicates that AQEC plays a profound role in the physics of CFT and gravity that could be fruitful to further study.  

Remarkably, the threshold of our complexity results' applicability, $\varDelta = D$, is of special physical significance as global symmetry current operators  have scaling dimension $D$. 
Therefore, if the system has a global symmetry, $\varDelta$ is capped at $D$ so that the CFT code obeys $\varepsilon=\Omega(1/n)$, thereby falling outside the regime of universal nontrivial complexity. As an implication, AQEC parameters could be used to probe symmetries.
This is consistent with the scaling limit of the AQEC error of covariant codes \cite{faist20,Woods2020continuousgroupsof,Zhou2021newperspectives,liu2022quantum} and
highlights the significance of the $\sim 1/n$ boundary from yet another angle.  
In addition, $\varDelta>D$ implies a fundamental tension with global symmetries on both sides of AdS/CFT  \cite{harlow2021symmetries}, which could be a situation of special significance given the key role of symmetries in quantum gravity \cite{Misner1957,PhysRevD.83.084019,harlow2021symmetries}.  
All things considered, our discussions suggest compelling motivations to look for and study CFTs in the $\varDelta>D$  regime.  This problem is nontrivial and interesting by itself because such a theory is not forbidden by known constraints in any $D$ (note that the stress-energy tensor has scaling dimension $D+1$), but we have limited knowledge of natural examples.

\section{Discussion and outlook}
\label{sec:discussion}
We studied general quantum code subspaces that do not necessarily enable exact QEC, offering a systematic understanding of their nontriviality.
We proved code error thresholds below which circuit complexity lower bounds for any code state remain robust, for any geometry.
Through the examples of topological order and critical systems that respectively represent gapped and gapless order, we have demonstrated that AQEC provides a powerful unifying lens for understanding the physics of complex many-body quantum systems, highlighting the value of insights from quantum information in physics.  
 %Our study reveals particularly sharp boundaries of AQEC error.

A noteworthy phenomenon is that intrinsically approximate codes with power-law error scaling, although highly atypical in the sense that randomly constructed codes  almost always exhibit exponentially small errors, naturally arise 
in wide-ranging physical scenarios like gapless systems in fundamental ways.
In particular, error scaling near the $\sim 1/n$ boundary often emerges hand in hand with symmetries or specific structures of states.  Based on our theory, the nontrivial order associated with such scenarios, which we call \emph{marginal order}, is expected to represent a general type of order that is fundamentally distinct from topological order in its stability and other physical properties. It would be interesting to further investigate this notion.  Note that the approximate Eastin--Knill theorems \cite{faist20,Woods2020continuousgroupsof,Zhou2021newperspectives,PhysRevLett.126.150503,PhysRevResearch.4.023107,liu2022quantum,liu2022approximate} place codes with continuous transversal gates or symmetries outside the acceptable code regime,  further strengthening the notion of incompatibility between symmetry and QEC.   

Overall, the $\sim 1/n$ boundary is worth further study. Note that the the approximate Eastin--Knill theorems \cite{faist20,Woods2020continuousgroupsof,Zhou2021newperspectives,PhysRevLett.126.150503,PhysRevResearch.4.023107,liu2022quantum,liu2022approximate} place codes with continuous transversal gates or symmetries outside the acceptable code regime,  further strengthening the notion of incompatibility between transversality and QEC.

The CFT codes represent  a family of intrinsically AQEC codes that warrant deeper investigation, especially because of their importance to the understanding of  quantum criticality as well as quantum gravity.
Specifically, our preliminary discussion  point towards several interesting avenues for more rigorous consideration, such as the existence of nontrivial gravity duals and implications for symmetries in quantum gravity. Recall also that our theory suggests particularly sharp properties of
$\varDelta > D$ theories, which would be interesting to explore.

Furthermore, it would be valuable to extend our theory to continuous variable and fermionic systems, mixed states, and more general QEC settings.

To conclude, a key takeaway is that we expect AQEC codes to substantially extend the scope and utility of conventional notions of QEC in the realms of practical quantum technologies, complexity theory, and physics. We hope this study sparks further exploration in AQEC and its physical as well as practical applications.

\begin{acknowledgments}
We thank Anurag Anshu, Yin-Chen He, Han Ma, Shengqi Sang, Beni Yoshida, Sisi Zhou, Zheng Zhou, Yijian Zou for valuable discussions and feedback. J.Y. would like to thank Anton Burkov for his  support. 
D.G.\ is partially supported by the National Science Foundation (RQS QLCI grant OMA-2120757).
Z.-W.L.\ is partially supported in part by a startup funding from YMSC, Tsinghua University, and NSFC under Grant No.~12475023.
Research at Perimeter Institute is supported in part by the Government of Canada through the Department of Innovation, Science and Economic Development Canada and by the Province of Ontario through the Ministry of Colleges and Universities.
\end{acknowledgments}

\bibliography{qec}

\onecolumngrid
\appendix

\section{Approximate quantum error correction and subsystem variance}
\label{app:aqec}

In this section, we provide detailed information for the preliminaries of AQEC  briefly introduced in Sec.~\ref{sec:prelim} of the main text. Specifically, we will formally define two  different types of quantitative measures of deviation from exact QEC codes---namely, QEC inaccuracy and subsystem variance---and examine their relationship. %Along the way, we formalize a general notion of the ``worst'' kind of noise channels that we call  complete noise channels,  which may be of independent interest.

The key point here is that we do not require the quantum code to encode logical information in such an ideal manner that it can be perfectly recovered (or decoded) after the system undergoes noise. Instead, we permit the process to be approximate. 
The degree of approximation  is naturally characterized by how well the recovery procedure can restore the logical information.  More precisely, we consider the optimal distance between the effective channel of the entire coding procedure and the logical identity channel which represents the situation where no logical information is lost.  Indeed, if a recovery map exists such that this distance is zero, that is, the logical information can be perfectly recovered, we say that exact QEC has been achieved. The larger this distance is, the further the procedure deviates from perfect recovery.

We now define the specific metrics that will be used. In the following, $\onenorm{O}$ denotes the trace norm of operator $O$ given by $\onenorm{O} := \tr\sqrt{O^\dagger O}$.
The purified distance between quantum states $\rho$ and $\sigma$ is defined as
\begin{equation}
    P(\rho, \sigma)\coloneqq\sqrt{1-f(\rho, \sigma)^2},
\end{equation}
where $f$ is the Uhlmann fidelity
\begin{equation}
    f(\rho, \sigma)\coloneqq\|\sqrt{\rho} \sqrt{\sigma}\|_1=\operatorname{Tr} \sqrt{\sqrt{\rho} \sigma \sqrt{\rho}}.
\end{equation}
Then the proper channel version of the purified distance, known as the completely bounded  purified distance, between two quantum channels $\mathcal{M}_1$ and $\mathcal{M}_2$, is defined as 
\begin{equation}
    P(\mathcal{M}_1, \mathcal{M}_2)\coloneqq\sqrt{1-F(\mathcal{M}_1, \mathcal{M}_2)^2},
\end{equation}
where $F$ is the completely bounded fidelity of channels given by
\begin{equation}
    F(\mathcal{M}_1, \mathcal{M}_2)\coloneqq\min_\rho f((\mathcal{M}_1 \otimes \id)(\rho),(\mathcal{M}_2 \otimes \id)(\rho)),
\end{equation}
with the optimization running over input states on any extended system.

\begin{defn}[QEC inaccuracy]
For encoding map $\mathcal{E}$ and noise channel $\mathcal{N}$ acting on the physical system, the \emph{QEC inaccuracy} is defined as 
    \begin{equation}
\tilde\varepsilon\left(\mathcal{N}, \mathcal{E}\right)\coloneqq\min _{\mathcal{R}} P\left(\mathcal{R} \circ \mathcal{N} \circ \mathcal{E}, \id_L\right),
\end{equation}
where $P$ is the channel purified distance and $\id_L$ denotes the logical identity channel.
\end{defn}
This quantity can be characterized using the complementary channel formalism  \cite{BO10:AKL}. %For erasure error, the complementary channel is a constant (replacement) channel $\mathcal{T}_\zeta^{L \rightarrow E}\left(\rho^L\right)=(\tr\rho^L) \zeta^E$.
Indeed, one can always view  ${ \mathcal { N } \circ \mathcal { E } }$ as a unitary mapping from the logical system $L$ to a joint system of the physical system $S$
and an environment system $E$ followed by tracing out $E$. The complementary channel of ${ \mathcal { N } \circ \mathcal { E } } $ given by tracing out $S$ instead, denoted
by $\widehat { \mathcal { N } \circ \mathcal { E } }$, describes the leakage of logical information due to the noise channel. Intuitively, if this complementary channel is a constant channel, i.e.,~it is independent of the logical state,  no logical information is leaked to the environmen, thus recovery on the physical system is possible; on the other hand, if the complementary channel is dependent on the logical state, it indicates that some logical information is leaked, which prohibits perfect recovery.  
A precise characterization is given by
\begin{equation}\label{eq:complementary-worst}
    \tilde\varepsilon\left(\mathcal{N}, \mathcal{E}\right)=\min _\zeta \max _{R,\ket{\psi}^{L R}} P\left(( \widehat { \mathcal { N } \circ \mathcal { E } } { } ^ { L \rightarrow E } \otimes \mathrm { id } ^ { R } ) (\ketbra{\psi}{\psi}^{L R}),(\mathcal{T}_\zeta^{L \rightarrow E} \otimes \mathrm{id}^R)(\ketbra{\psi}{\psi}^{L R})\right).
\end{equation}
where $R$ is a reference system, and $\mathcal{T}_\zeta(\rho) = \tr(\rho)\zeta$ is a constant channel that outputs state $\zeta$ in the environment.

We will also use the following adapted form of \cite[Cor.\ 2]{BO10:AKL} (see also Ref.~\cite{Ng10}).
\begin{lem}
\label{lemma:KL-general}
Let $\Pi$ be the projector onto the code subspace of an isometric quantum code, and consider noise channel $\mathcal{N}_S(\cdot) = \sum_\alpha K_\alpha(\cdot)K_\alpha^\dagger$. %An isometric quantum code defined by projector $\Pi$ is  $\varepsilon$-correctable under $\mathcal{N}$ if and only if
    %\begin{equation}
       %\Pi K_i^\dagger K_j \Pi = \lambda_{ij} \Pi + \Pi {B }_{ij} \Pi,
    %\end{equation}
    %where $\lambda_{ij}$ are the components of a density operator, and $P(\Lambda+\mathcal{B},\Lambda)\leq \varepsilon$ where $\Lambda(\rho)=\sum_{ij}\lambda_{ij}\Tr(\rho)|i\rangle\langle j|$ and $(\Lambda+\mathcal{B})(\rho)=\Lambda(\rho)+\sum_{ij}\Tr(\rho B_{ij})|i\rangle\langle j|$.
    Then 
\begin{equation}
\tilde\varepsilon(\mathcal{N},\mathcal{E}) = \min_{{\Lambda}} {P({\Lambda},{\Lambda}+{\mathcal{B}})}, 
\end{equation}
where ${\Lambda}(\rho) = \sum_{\alpha\beta} \lambda_{\alpha\beta} \Tr(\rho) \ket{\alpha}\bra{\beta}$, $({\Lambda}+{\mathcal{B}})(\rho) = {\Lambda}(\rho) + \sum_{\alpha\beta} \Tr(\mathcal{E}(\rho){B }_{\alpha\beta} ) \ket{\alpha}\bra{\beta}$, with $\ket{\alpha},\ket{\beta}$ being orthonormal basis states, and $\lambda_{\alpha\beta}$ and ${B }_{\alpha\beta}$ are constant numbers and operators satisfying 
\begin{equation}
    \Pi K_\alpha^\dagger K_\beta \Pi = \lambda_{\alpha\beta} \Pi + \Pi {B }_{\alpha\beta} \Pi.  \label{eq:kl-approx}
\end{equation}
\end{lem}

From the definition, it can be observed that $\mathcal{B}$ essentially characterizes the violation of the Knill--Laflamme conditions.

It is important to note that the QEC inaccuracy is defined with respect to a specific noise channel and the corresponding  recovery maps. 
What intrinsic properties of the code is it related to? To answer this, we introduce the following type of quantities.

\begin{defn}[Subsystem variance]
Given a code subspace $\mathfrak{C}$ (i.e., the image of the encoding map $\mathcal{E}$), consider a reference state given by the maximally mixed state of  $\mathfrak{C}$, namely, the equal statistical mixture of a set of basis states $\{\ket{\psi_i}\}$ that span $\mathfrak{C}$,
\begin{equation}
\label{eq:ref_state}
    \Gamma\coloneqq\frac{1}{2^k} \sum_{i=1}^{2^k}   |\psi_i\rangle\langle\psi_i|.
\end{equation}
Consider subsystems that are connected (local) with respect to a given adjacency graph $\mathsf{G}$ associated with the physical system.
The \emph{subsystem variance} for a particular subsystem $R$ is defined as
    \begin{align}
            \varepsilon_\mathsf{G}(\mathfrak{C},R)& \coloneqq \max_{\sigma\in\mathfrak{C}}\onenorm{\sigma_R - \Gamma_R}.
       \end{align}
       where  $\sigma=\mathcal{E}(\rho)$ is a code state (subscript $R$ denotes the reduced state on $R$).
       By convexity, the maximum is necessarily attained by some pure state.
       Then the \emph{overall subsystem variance} for size-$d$ subsystems is 
       \begin{align}
\varepsilon_\mathsf{G}(\mathfrak{C},d)&\coloneqq\max_{|R|\leq d}\varepsilon_\mathsf{G}(\mathfrak{C},R).
    \end{align}
%where $\Gamma\coloneqq\frac{1}{2^k} \sum_{i=1}^{2^k}   |\psi_i\rangle\langle\psi_i|$  is the statistical average of $\{\psi_i\}$ that span $\mathfrak{C}$, namely the maximally mixed state of $\mathfrak{C}$, and $S$ is a connected (local) subsystem with respect to the adjacency graph $G$ (subscript $S$ denotes the reduced state on $S$).  
\end{defn}

To gain some intuition, observe that zero subsystem variance implies that all code states are completely locally indistinguishable, thus the subsystem does not carry any logical information. As a result, noise actions within the subsystem cannot cause any  leakage of logical information into the environment and are therefore exactly recoverable.   On the other hand, if the subsystem variance is nonzero, then local noise actions can cause logical information leakage, prohibiting exact QEC.

Now let us discuss the general relationship between QEC inaccuracy and subsystem variance in a quantitative sense.  %Recall that the QEC inaccuracy aims to characterize the best possible recovery accuracy under certain noise channels, while the subsystem variance is an intrinsic property of the code. %Therefore, their relationship relies on the noise channel under consideration.  
In light of Lemma~\ref{lemma:KL-general}, we need to understand the connection between $\mathcal{B}$ and subsystem variance.  We now prove the following equivalence relation for noise effects represented by replacement channels, i.e., channels that outputs a fixed state independent of the input, the most important cases of which are erasure noise that outputs some garbage state, and completely depolarizing noise that outputs the maximally mixed state.
\begin{prop}\label{thm:replacement_cond}
Consider a noise channel $\mathcal{N}$ given by a replacement channel acting nontrivially on subsystem $R$, i.e., $ {\mathcal{N}}(\sigma)=(\Tr_{R}\sigma)_{\ols R}\otimes\gamma$ for a fixed state $\gamma$ ($\ols R$ denotes the complement of $R$). {Taking $\lambda_{\alpha\beta}=\Tr(\Gamma K_{\alpha}^\dagger K_\beta)$ in Lemma~\ref{lemma:KL-general}, } for any code and any state $\rho$, we have 
\begin{equation}
         \onenorm{\mathcal{B}(\rho)}= \onenorm{\sigma_R-\Gamma_R}.\label{eq:B=delta}
\end{equation}
\end{prop}
\begin{proof}
    Recall the definition of $\mathcal{B}(\rho)$,
\begin{equation}
\label{eq:defB}
    \mathcal{B}(\rho)=\sum_{\alpha\beta}\Tr(\mathcal{E}(\rho)B_{\alpha\beta})|\alpha\rangle\langle \beta|,
\end{equation}
where
\begin{equation}\label{eq:B }
    \Pi {B }_{\alpha\beta} \Pi=\Pi K_\alpha^\dagger K_\beta \Pi - \lambda_{\alpha\beta} \Pi, 
\end{equation}
with the Kraus operators associated with the noise channel $ {\mathcal{N}}(\sigma)=(\Tr_{R}\sigma)_{\ols R}\otimes\gamma$. Suppose the spectral decomposition of $\gamma$ is given by $\gamma=\sum_{\alpha}c_\alpha |\alpha\rangle\langle\alpha|$. Then the Kraus operators can be taken  to be $K_{\alpha i}=\sqrt{c_\alpha}|\alpha\rangle\langle i_R|\otimes\mathbbm{1}_{\ols{R}}$ where $\{\ket{i_R}\}$ is a basis of $\mathcal{H}_R$. So we obtain
\begin{align}
      \mathcal{B}(\rho)_{\alpha i, \beta j}&=\Tr(\sigma B_{\alpha i, \beta j})\\&=\Tr(\sigma \Pi B_{\alpha i, \beta j}\Pi)\\
      &=\Tr(\sigma \Pi K_{\alpha i}^\dagger K_{\beta j} \Pi)-\Tr(\sigma \lambda_{\alpha i,\beta j}\Pi)\\
      &=\Tr(\sigma K_{\alpha i}^\dagger K_{\beta j})-\Tr(\Gamma K_{\alpha i}^\dagger K_{\beta j})\\
      &=\delta_{\alpha\beta} c_\alpha(\sigma_R-\Gamma_R)_{{j_R}{i_R}},
      \label{eq:B_ab}
\end{align}
where $\sigma_R$ and $\Gamma_R$ are, respectively, the reduced density matrices of $\sigma$ and $\Gamma$ on subsystem $R$. 
We have repeatedly used the code state property $\Pi\psi\Pi = \psi$ and the cyclic property of the trace function. 
In the third line, we used (\ref{eq:B }).
In the fourth line, we used $\lambda_{\alpha\beta}=\Tr(\Gamma K_\alpha^\dagger K_\beta)$.
From this, we can see that $\mathcal{B}(\rho)$ is a block-diagonal matrix with each block being $(\sigma_R-\Gamma_R)^T$, so we have
\begin{equation}
    \onenorm{\mathcal{B}(\rho)}=\sum_\alpha c_\alpha\onenorm{(\sigma_R-\Gamma_R)^T}
    =\onenorm{\sigma_R-\Gamma_R},
\end{equation}
where we used $\sum_\alpha c_\alpha =1$.
\end{proof}

For general channels, {with $\lambda_{\alpha\beta}$ taken in the same way}, we have
$       \onenorm{\mathcal{B}(\rho)}\leq\onenorm{\sigma_R-\Gamma_R}$.
This can be shown by observing that for a general noise channel $\mathcal{N} = \check{\mathcal{N}}_R\otimes\mathrm{id}_{\ols R}$ acting on a subsystem $R$ with $\check{\mathcal{N}}_R=\sum_{\alpha} \check{K}_{\alpha}(\cdot)\check{K}_\alpha^\dagger$, we can define the following linear transformation %\jinmin{motivated by (\ref{eq:defB})}
\begin{equation}
    \mathcal{E}_{\mathcal{B}}(\sigma)\coloneqq\sum_{\alpha,\beta}\Tr(\check{K}_\beta\sigma\check{K}_\alpha^\dagger)|\alpha\rangle\langle\beta|,
 \end{equation}
 {One can see that if we take $\lambda_{\alpha\beta}=\Tr(\Gamma K_{\alpha}^\dagger K_\beta)$ in Lemma~\ref{lemma:KL-general}, $\mathcal{B}(\rho)=\mathcal{E}_{\mathcal{B}}(\sigma_R-\Gamma_R)$. Since $\mathcal{E}_{\mathcal{B}}$ can be understood as}
 a transformation of the form $\mathcal{T}\circ \widehat{\mathcal{N}}$, where $\mathcal{T}$ denotes matrix transposition and $\widehat{\mathcal{N}}$ is a complementary channel  of the noise channel $\mathcal{N}$. Then the claim follows from the fact that the 1-norm is preserved under transposition and monotonic under quantum channels.  
The replacement channels represent extreme cases of noise channels that saturate this bound, that is, they induce the worst violation of the Knill--Laflamme conditions. Indeed, replacement channels completely destroy the logical information within their range and, as such, should set upper bounds on the effects of general noise channels.

With these setups in place, we now prove a two-way bound for QEC inaccuracy  in terms of subsystem variance, as claimed in 
(\ref{eq:error_variance}) in the main text.  We first introduce a technical lemma about error metrics that will be used in the proof, and then proceed to state and prove the result.
%Note that, if for an arbitrary noise channel we can 

\begin{lem}\label{lem:FvdG-for-B}
Denote $\onenorm{\mathcal{B}}\coloneqq \max_\rho\onenorm{\mathcal{B}(\rho)}.$ 
For any $\Lambda$ it holds that
\begin{equation}
    \frac{1}{2}\onenorm{\mathcal{B}}\leq P(\Lambda,\Lambda+\mathcal{B})\leq2^{k/2}\sqrt{\onenorm{\mathcal{B}}}.\label{eq:fvdg-b-p}
\end{equation}
\end{lem}
\begin{proof}
%\ziwen{//}
    Recall the Fuchs–van de Graaf inequalities,
\begin{equation}\label{eq:FvdG}
    1-f(\rho,\sigma)\leq \frac{1}{2}\onenorm{\rho-\sigma}\leq \sqrt{1-f(\rho,\sigma)^2}.
\end{equation}
Upper bound in (\ref{eq:fvdg-b-p}):  Note that
\begin{align}
%\begin{split}
    F(\Lambda,\Lambda+\mathcal{B})&=\min_\rho f(\Lambda\otimes \id(\rho),(\Lambda+\mathcal{B})\otimes \id(\rho))\\
    &\geq 1-\frac{1}{2}\max_\rho\onenorm{\mathcal{B}\otimes \id(\rho)}\\
    %&=1-\frac{1}{2}\|\mathcal{B}\|^{(2^k)}_1\\
    %&=1-\frac{1}{2}\max_\rho\|\mathcal{B}\otimes \id_{2^k}(\rho)\|_1\\
    &\geq 1-2^{k-1} \onenorm{\mathcal{B}},
%\end{split}
\end{align}
where the second line follows from (\ref{eq:FvdG})  and the last line is taken from e.g.~Refs.~\cite{HaydenWinter10,BCSB19:aqec}.
%\begin{equation}
    %\|\Sigma\|^1\leq\|\Sigma\|^{(2^k)}\leq2^k \|\Sigma\|^1
%\end{equation}
So we have
\begin{align}
%\begin{split}
    P(\Lambda,\Lambda+\mathcal{B})&\leq\sqrt{1-(1-2^{k-1} \onenorm{\mathcal{B}})^2}=\sqrt{2^{k} \onenorm{\mathcal{B}}-2^{2k-2} \onenorm{\mathcal{B}}^2}\leq2^{k/2}\sqrt{\onenorm{\mathcal{B}}}.
%\end{split}
\end{align}
% Let's try to reverse the inequalities in the following. 
Lower bound in  (\ref{eq:fvdg-b-p}): Note first that
\begin{equation}
    P(\Lambda,\Lambda+\mathcal{B})=\max_\rho P(\Lambda\otimes \id(\rho),(\Lambda+\mathcal{B})\otimes \id(\rho)).
\end{equation}
By (\ref{eq:FvdG}),
\begin{align}
%\begin{split}
   & P(\Lambda\otimes \id(\rho),(\Lambda+\mathcal{B})\otimes \id(\rho)) %&\geq\frac{1}{2}\onenorm{\Lambda\otimes \id(\rho)-(\Lambda+\mathcal{B})\otimes \id(\rho)}\\
    \geq \frac{1}{2}\onenorm{\mathcal{B}\otimes \id(\rho)}.
%\end{split}
\end{align}
Taking maximization on both sides, we obtain
\begin{equation}
    P(\Lambda,\Lambda+\mathcal{B})\geq\frac{1}{2}\max_\rho\onenorm{\mathcal{B}\otimes \id(\rho)}\geq\frac{1}{2}\onenorm{\mathcal{B}}.
\end{equation}
%\ziwen{Choi}
\end{proof}

\begin{prop}\label{thm:two-way}
Let $\check{\mathcal{N}}_R$ be any replacement channel acting on a $d$-qubit subsystem $R$ that is connected with respect to adjacency graph $\mathsf{G}$. Denote the overall channel by $\mathcal{N} = \check{\mathcal{N}}_R\otimes\mathrm{id}_{\ols R}$.  It holds that
    \begin{equation}
    \frac{1}{4}\varepsilon_\mathsf{G}(\mathfrak{C},R) \leq \tilde\varepsilon(\mathcal{N}, \mathcal{E}) \leq 2^{k/2}\sqrt{\varepsilon_\mathsf{G}(\mathfrak{C},R)}.
\end{equation}
A version for the overall subsystem variance  is directly obtained by optimizing over $R$: 
    \begin{equation}
    \label{eq:app_error_variance}
    \frac{1}{4}\varepsilon_\mathsf{G}(\mathfrak{C},d) \leq \max_{R}\tilde\varepsilon(\mathcal{N}, \mathcal{E}) \leq 2^{k/2}\sqrt{\varepsilon_\mathsf{G}(\mathfrak{C},d)}.
\end{equation}
\end{prop}
\begin{proof}
    Following arguments in Ref.~\cite{BO10:AKL}, for noise channel $\mathcal{N}(\cdot) = \sum K_\alpha(\cdot)K_\beta^\dagger$, let
\begin{equation}
\Lambda(\rho)=\sum_{ij}\Tr(K_i^\dagger K_j\Gamma) \Tr(\rho)|i\rangle\langle j|,
\end{equation}
in Lemma~\ref{lemma:KL-general}, then we have 
\begin{equation}
    P(\Lambda, \Lambda+\mathcal{B})\leq 2\tilde\varepsilon(\mathcal{N}, \mathcal{E}).\label{eq:2eps}
\end{equation}
  This can be seen as follows.  Using the language of Ref.~\cite{BO10:AKL},  for the case of our interest $\mathcal{M} = \id$, such a $\Lambda$  corresponds to  $\widehat{\mathcal{N}\circ\mathcal{E}}\circ\widehat{\mathcal{M}}$ where $\widehat{\mathcal{M}}$ outputs $\Gamma$ and  $\Lambda+\mathcal{B}$  corresponds to $\widehat{\mathcal{N}\circ\mathcal{E}}$. By triangle inequality we have
\begin{equation}
    P(\widehat{\mathcal{N}\circ\mathcal{E}},\widehat{\mathcal{N}\circ\mathcal{E}}\circ\widehat{\mathcal{M}}) \leq P(\widehat{\mathcal{N}\circ\mathcal{E}},\widetilde{\mathcal{R}}\circ\widehat{\mathcal{M}}) + P(\widehat{\mathcal{N}\circ\mathcal{E}}\circ\widehat{\mathcal{M}},\widetilde{\mathcal{R}}\circ\widehat{\mathcal{M}}),
\end{equation} where $\widetilde{R}$ is the optimal channel that achieves $\tilde\varepsilon(\mathcal{N}, \mathcal{E})$ so  $P(\widehat{\mathcal{N}\circ\mathcal{E}},\widetilde{\mathcal{R}}\circ\widehat{\mathcal{M}}) = \tilde\varepsilon(\mathcal{N}, \mathcal{E})$, and 
$P(\widehat{\mathcal{N}\circ\mathcal{E}}\circ\widehat{\mathcal{M}},\widetilde{\mathcal{R}}\circ\widehat{\mathcal{M}}) = P(\widehat{\mathcal{N}\circ\mathcal{E}}\circ\widehat{\mathcal{M}},\widetilde{\mathcal{R}}\circ\widehat{\mathcal{M}}^2) \leq P(\widehat{\mathcal{N}\circ\mathcal{E}},\widetilde{\mathcal{R}}\circ\widehat{\mathcal{M}}) = \tilde\varepsilon(\mathcal{N}, \mathcal{E})$
 where we used $\widehat{\mathcal{M}}^2 = \widehat{\mathcal{M}}$. Therefore, $P(\widehat{\mathcal{N}\circ\mathcal{E}},\widehat{\mathcal{N}\circ\mathcal{E}}\circ\widehat{\mathcal{M}}) \leq 2\tilde\varepsilon(\mathcal{N}, \mathcal{E})$.  Note that this holds more generally in operator algebra QEC, beyond the standard  $\mathcal{M} = \id$ setting.
 Consider a replacement channel $\mathcal{N}_R$ acting on a certain $R$.
For any code state ${\sigma}$  we have
\begin{equation}\label{eq:delta_psi_upper}
    \onenorm{\sigma_R-\Gamma_R}=\onenorm{\mathcal{B}(\rho)}\leq\onenorm{\mathcal{B}}\leq2P(\Lambda,\Lambda+\mathcal{B})\leq  4\tilde\varepsilon(\mathcal{N}, \mathcal{E}),
\end{equation}
where the first step is due to Proposition~\ref{thm:replacement_cond}, the third step follows from the lower bound in Lemma~\ref{lem:FvdG-for-B}, and the last step follows from (\ref{eq:2eps}).
Therefore, by definition, $\varepsilon_\mathsf{G}(\mathfrak{C},R) \leq 4\tilde\varepsilon(\mathcal{N}, \mathcal{E})$.

On the other hand, observe that the upper bound in Lemma~\ref{lem:FvdG-for-B}, namely the lower bound on $\onenorm{\mathcal{B}}$, can be translated to a lower bound on the subsystem variance.  Explicitly, consider a replacement channel $\check{\mathcal{N}}_R$ acting on a certain $R$. It holds that
\begin{align}
\max_{\sigma}\onenorm{\sigma_R-\Gamma_R}=\max_{\rho}\onenorm{\mathcal{B}(\rho)}=\onenorm{\mathcal{B}} \geq \left(\frac{P(\Lambda,\Lambda+\mathcal{B})}{2^{k/2}}\right)^2 \geq  2^{-k}\tilde\varepsilon(\mathcal{N}, \mathcal{E})^2,
\end{align}
where the first step is again due to Proposition~\ref{thm:replacement_cond}, the second step follows from the convexity of 1-norm, the third step follows from the upper bound in Lemma~\ref{lem:FvdG-for-B}, and the last step follows from Lemma~\ref{lemma:KL-general}. This indicates $\tilde\varepsilon(\mathcal{N}, \mathcal{E}) \leq 2^{k/2}\sqrt{\varepsilon_\mathsf{G}(\mathfrak{C},R)}$.
\end{proof}
More generally, for an arbitrary noise channel, by understanding  the relationship between the $  \onenorm{\mathcal{B}(\rho)}$ associated with it and $\onenorm{\sigma_R-\Gamma_R}$, we can derive analogous bounds for $\tilde\varepsilon$.  This may be useful in certain scenarios such as channel-adapted QEC settings.

\section{Coherent information and subsystem variance}\label{sec:coh-info}

In this section, we establish a general set of two-way bounds that quantitatively relate coherent information under noise channels with subsystem variance. 

The coherent information associated with the code subspace $\mathfrak{C}$ and the noise channel $\mathcal{N}$ is a well-known quantity in quantum information, which measures the amount of recoverable quantum information after $\mathcal{N}~$\cite{cohinfo,PhysRevA.55.1613,CoherinfoAQEC2001}:

\begin{defn}[Coherent information]
    Given a code subspace $\mathfrak{C}$ (i.e., the image of the encoding map $\mathcal{E}$) with a set of basis states $\{|\psi_i\rangle\}$ that span $\mathfrak{C}$, consider the input state
\begin{equation}
        \ket{\Phi_\mathfrak{C}}=\frac{1}{2^{k/2}}\sum_i|\psi_i\rangle_Q|i\rangle_R,
    \end{equation}
    which is maximally entangled between the code system $Q$ and the reference system $R$. 
    After the action of the noise channel $\mathcal{N}:Q\rightarrow Q'$, the output state is given by 
    \begin{equation}
        \rho_{Q'R}=\mathcal{N}_{Q\rightarrow Q'}\otimes \id_{R}(\ket{\Phi_\mathfrak{C}}\bra{\Phi_\mathfrak{C}}).
    \end{equation}
    Then the \emph{coherent information} associated with the code subspace $\mathfrak{C}$ and the noise channel $\mathcal{N}$ is defined as
    \begin{equation}
    I_c(\mathfrak{C};\mathcal{N})\coloneqq S^{Q'}-S^{RQ'}.
    \end{equation}
    where $S^{Q'}$ and $S^{RQ'}$ are the entanglement entropies of systems $Q'$ and $RQ'$, respectively.
\end{defn}

It has been proven that a noise channel $\mathcal{N}$ is exactly correctable if and only if the corresponding coherent information can take the maximum value $I_c(\mathfrak{C};\mathcal{N})=k$ ~\cite{CoherinfoAQEC2001}. In the AQEC context where $\mathcal{N}$ is not necessarily exactly correctable, it is natural to expect that the inaccuracy is related to the deviation from this maximum value.

To make this more rigorous, we prove the following two-way bounds for coherent information in terms of subsystem variance, thereby relating coherent information,  QEC inaccuracy, and subsystem variance along with the bounds between the latter two proven above in Sec.~\ref{app:aqec}. 
\begin{prop}\label{thm:coh-var-two-way}
    Let $\check{\mathcal{N}}_A$ be any replacement channel acting on a $d$-qubit subsystem $A$ that is connected with respect to adjacency graph $\mathsf{G}$. Denote the overall channel $\mathcal{N}:Q\rightarrow Q'$ by $\mathcal{N} = \check{\mathcal{N}}_A\otimes\mathrm{id}_{\ols A}$ ($\ols A$ denotes the complement of $A$ in $Q$). It holds that
    \begin{equation}
   2^{-2k-d-1} \varepsilon^2_\mathsf{G}(\mathfrak{C},A)\leq k-I_c(\mathfrak{C};\mathcal{N})\leq 2^{k}\log(2^{d+k}-1)\varepsilon_\mathsf{G}(\mathfrak{C},A)+H_2(2^{k}\varepsilon_\mathsf{G}(\mathfrak{C},A)).
\end{equation}
A version for the overall subsystem variance  is directly obtained by optimizing over $A$: 
\begin{equation}
   2^{-2k-d-1} \varepsilon^2_\mathsf{G}(\mathfrak{C},d)\leq k-\min_{A}I_c(\mathfrak{C};\mathcal{N})\leq 2^{k}\log(2^{d+k}-1)\varepsilon_\mathsf{G}(\mathfrak{C},d)+H_2(2^{k}\varepsilon_\mathsf{G}(\mathfrak{C},d)).
\end{equation}
\end{prop}

\begin{proof}
For the noise channel $\mathcal{N}: Q\rightarrow Q'$ given by replacement channel $\check{\mathcal{N}}_A$ acting nontrivially on subsystem $A$, we have $ {\mathcal{N}}(\sigma)=(\Tr_{A}\sigma)_{\ols A}\otimes\gamma_A$ for a fixed state $\gamma_A$. Hence, for the input state
    \begin{equation}
        \ket{\Phi_\mathfrak{C}}=\frac{1}{2^{k/2}}\sum_i|\psi_i\rangle_Q|i\rangle_R,
    \end{equation}
the output state is given by
\begin{equation}
    \rho_{Q'R}=\Tr_{A}(\ket{\Phi_\mathfrak{C}}\bra{\Phi_\mathfrak{C}})\otimes \gamma_A.
\end{equation}
The coherent information is given by
\begin{align}
    I_c(\mathfrak{C};\mathcal{N})&=(S^{\ols{A}}+S(\gamma_A))-(S^{\ols{A}R}+S(\gamma_A))=S^{\ols{A}}-S^{\ols{A}R}=S^{AR}-S^A.
\end{align}
So for any replacement noise acting nontrivially on subsystem $A$, the coherent information 
\begin{equation}
\label{eq:coh-mut-relation}
    I_c(\mathfrak{C};\mathcal{N})=k-I_{\Phi_\mathfrak{C}}( A:R),
\end{equation}
which is determined by the mutual information $I^\mathfrak{C}( A:R)$ between systems $A$ and $R$.
Using the quantum Pinsker
inequality, we can lower-bound the mutual information as
\begin{equation}
I_{\Phi_\mathfrak{C}}( A:R)=S(\rho_{AR}\|\rho_{A}\otimes\rho_{R})\geq \frac{1}{2}\onenorm{\rho_{AR}-\rho_{A}\otimes\rho_{R}}^2.
\end{equation}
Since  $\twonorm{A}\leq\onenorm{A}\leq\sqrt{n}\twonorm{A}$ for an $n\times n$ matrix $A$, we obtain
\begin{align}
    \onenorm{\rho_{AR}-\rho_{A}\otimes\rho_{R}}&\geq \twonorm{\rho_{AR}-\rho_{A}\otimes\rho_{R}}
    =\frac{1}{2^{k}}\twonorm{\sum_{ij}\delta \rho_{A}^{ij}\otimes|i\rangle_R\langle j|}
    =\frac{1}{2^{k}}\sqrt{\tr\sum_{ij}(\delta \rho_{A}^{ij})^\dagger \delta \rho_{A}^{ij}}\\
    &\geq\frac{1}{2^{k}}\sqrt{\tr\sum_{i}(\delta \rho_{A}^{ii})^\dagger \delta \rho_{A}^{ii}}
    =\frac{1}{2^{k}}\sqrt{\sum_{i}\twonorm{\delta \rho_{A}^{ii}}^2}\geq \frac{1}{2^{k+d/2}}\sqrt{\sum_{i}\onenorm{\delta \rho_{A}^{ii}}^2}\\
    \label{eq:ARsplit_lower}
    &\geq 2^{-k-d/2}\varepsilon_\mathsf{G}(\mathfrak{C},A)),
\end{align}
where for simplicity we used the notations 
\begin{align}
    \delta\rho_{A}^{ij}&\equiv\tr_{\ols{A}}|\psi_i\rangle_Q\langle\psi_j|-\Gamma_A \delta_{ij}\,,\\
    \Gamma_A&\equiv\frac{1}{2^k}\sum_i \tr_{\ols{A}}|\psi_i\rangle_Q\langle\psi_i|\,.
\end{align}
In the final inequality, we implicitly used the fact that the mutual information is not basis-dependent. Recall the definition of the subsystem variance, 
\begin{equation}
\varepsilon_\mathsf{G}(\mathfrak{C},A)=\max_{|\psi\rangle\in\mathfrak{C}}\onenorm{\Gamma_A-\tr_{\ols{A}}|\psi\rangle_Q\langle\psi|}.
\end{equation}
The basis independence then allows us to select a basis for the code space such that the maximum in the subsystem variance corresponds to one of the basis vectors.
In conclusion, by combining (\ref{eq:coh-mut-relation}) and (\ref{eq:ARsplit_lower}), we obtain 
\begin{equation}
    I_{\Phi_\mathfrak{C}}( A:R)\geq 2^{-2k-d-1} \varepsilon^2_\mathsf{G}(\mathfrak{C},A),
\end{equation}
that is,
\begin{equation}
    k-I_c(\mathfrak{C};\mathcal{N})\geq 2^{-2k-d-1} \varepsilon^2_\mathsf{G}(\mathfrak{C},A).
\end{equation}

On the other hand, we can use the continuity bounds for the entropies to obtain an upper bound for the mutual information as follows: 
\begin{align}
    I_{\Phi_\mathfrak{C}}(A:R)&=S(\rho_{A})+S(\rho_{R})-S(\rho_{AR})=|S(\rho_{A}\otimes\rho_{R})-S(\rho_{AR})|\\
    &\label{Eq:mut_fannes}
    \leq \frac{1}{2}\onenorm{\rho_{AR}-\rho_{A}\otimes\rho_{R}}\log(2^{d+k}-1)+H_2\left(\frac{1}{2}\onenorm{\rho_{AR}-\rho_{A}\otimes\rho_{R}}\right).
\end{align}
Note that
\begin{align}
    \onenorm{\rho_{AR}-\rho_{A}\otimes\rho_{R}}&=\frac{1}{2^{k}}\onenorm{\sum_{ij}\delta \rho_{A}^{ij}\otimes|i\rangle_R\langle j|}
    \\
    &\leq\frac{1}{2^{k}}\sum_{ij}\onenorm{\delta \rho_{A}^{ij} }\cdot\onenorm{|\psi_i\rangle_R\langle\psi_j|}\\
    &\leq 
    \label{Eq:ARsplit_upper}
    \frac{1}{2^{k}}\cdot(2^{k})^2\cdot 2\varepsilon_\mathsf{G}(\mathfrak{C},A) =2^{k+1}\varepsilon_\mathsf{G}(\mathfrak{C},A),
\end{align}
where we used the fact that $\onenorm{\delta \rho_{A}^{ij}}\leq 2\varepsilon_\mathsf{G}(\mathfrak{C},A)$. For $i=j$, this follows from the definition of the subsystem variance. For $i\neq j$, this is proven by applying
$
\onenorm{\Gamma_A-\tr_{\ols{A}}|\psi\rangle_Q\langle\psi|}\leq\varepsilon_\mathsf{G}(\mathfrak{C},A)
$ 
to the state $\frac{1}{\sqrt{2}}(\ket{\psi_i}\pm\ket{\psi_j})$ and $\frac{1}{\sqrt{2}}(\ket{\psi_i}\pm i\ket{\psi_j})$, from which we get
\begin{align}
    &\frac{1}{2}\onenorm{\delta \rho_{A}^{ii}+\delta \rho_{A}^{ij}+\delta \rho_{A}^{ji}+\delta \rho_{A}^{jj}}\leq\varepsilon_\mathsf{G}(\mathfrak{C},A),\\
    &\frac{1}{2}\onenorm{\delta \rho_{A}^{ii}-\delta \rho_{A}^{ij}-\delta \rho_{A}^{ji}+\delta \rho_{A}^{jj}}\leq\varepsilon_\mathsf{G}(\mathfrak{C},A),\\
    &\frac{1}{2}\onenorm{\delta \rho_{A}^{ii}+i\delta \rho_{A}^{ij}-i\delta \rho_{A}^{ji}+\delta \rho_{A}^{jj}}\leq\varepsilon_\mathsf{G}(\mathfrak{C},A),\\
    &\frac{1}{2}\onenorm{\delta \rho_{A}^{ii}-i\delta \rho_{A}^{ij}+i\delta \rho_{A}^{ji}+\delta \rho_{A}^{jj}}\leq\varepsilon_\mathsf{G}(\mathfrak{C},A).
\end{align}
Then  we  obtain $\onenorm{\delta \rho_{A}^{ij}}\leq 2\varepsilon_\mathsf{G}(\mathfrak{C},A)$ by using the triangle inequality.
Therefore, by combining (\ref{Eq:mut_fannes}) and (\ref{Eq:ARsplit_upper}), we obtain
\begin{equation}
    I_{\Phi_\mathfrak{C}}(A:R)\leq 2^{k} \log(2^{d+k}-1)\varepsilon_\mathsf{G}(\mathfrak{C},A)+H_2(2^{k}\varepsilon_\mathsf{G}(\mathfrak{C},A)).
\end{equation}
Equivalently, we have
\begin{equation}
    k-I_c(\mathfrak{C};\mathcal{N})\leq 2^{k} \log(2^{d+k}-1)\varepsilon_\mathsf{G}(\mathfrak{C},A)+H_2(2^{k}\varepsilon_\mathsf{G}(\mathfrak{C},A)).
\end{equation}
by (\ref{eq:coh-mut-relation}).
To summarize, we have established the following relationship between coherent information and subsystem variance,
    \begin{equation}
   2^{-2k-d-1} \varepsilon^2_\mathsf{G}(\mathfrak{C},A)\leq k-I_c(\mathfrak{C};\mathcal{N})\leq 2^{k}\log(2^{d+k}-1)\varepsilon_\mathsf{G}(\mathfrak{C},A)+H_2(2^{k}\varepsilon_\mathsf{G}(\mathfrak{C},A).
\end{equation}
\end{proof}
As in Sec.~\ref{app:aqec}, the replacement channel represents a canonical setting and it is possible to derive analogous bounds for general noise channels.

\section{Detailed proofs and extensions of the circuit complexity results}
\label{app:proof}
In this section, we provide the complete forms of the complexity theorems summarized in the main text, along with their detailed proofs.  
%This section is structured as follows.
We will first present the detailed results and proof for the all-to-all case without geometric locality,  which corresponds to
 Theorem~\ref{thm:all_comp}, and then those for the slightly more technical finite-dimensional case that has geometric locality, which corresponds to
Theorem~\ref{thm:geo_comp}. Finally, we will formulate a general-form result that is applicable to arbitrary adjacency graphs.
\begin{thm}
\label{thm:app_all_comp}
Given an $(\!(n,k)\!)$ code $\mathfrak{C}$,
the $\delta$-robust all-to-all quantum circuit complexity of any code state $\ket{\psi}\in\mathfrak{C}$ satisfies 
\begin{equation}
    \cc^{\delta}(\psi)>\log d,
\end{equation}
if for complete graph $\mathsf{g}$,
\begin{equation}
    H_2(\varepsilon_\mathsf{g}(\mathfrak{C},d)/2 + \delta/2) < k/n,
\end{equation} 
or \begin{equation}
    H_2(2\max_{\mathcal{N}}\tilde\varepsilon(\mathcal{N},\mathcal{E}) + \delta/2) < k/n,
\end{equation}
where $\mathcal{N}$ is any $d$-qubit replacement channel and the arguments of $H_2$ are assumed to be less than $1/2$.
\end{thm}
This is proven by adapting ideas in Ref.~\cite{AnshuNirkhe20} to our scenario.  We first introduce a general result Proposition~\ref{lem:app_all_comp},  and then apply it to code spaces and take robustness (approximation tolerance) into account to obtain the theorem statements.  Note that this general result is independently useful for proving circuit complexity bounds for specific states. As mentioned, the propositions are useful for proving circuit complexity bounds for certain specific states that are not covered by the theorems.
\begin{prop}\label{lem:app_all_comp}
   Let $\psi = \ketbra{\psi}{\psi}$ be a $n$-qubit state and let
\begin{equation}
\Gamma \coloneqq \frac{1}{M} \sum_{i=1}^{M}|\phi_i\rangle\langle\phi_i|
\end{equation}
be a reference state where the $\ket{\phi_i}$'s are $n$-qubit states orthogonal to each other (namely, $\Gamma$ is a maximally mixed state of the space  spanned by basis $\{\ket{\phi_i}\}$).  
If for any set of qubits $R$ with $|R|\leq d$, it holds that 
\begin{equation}\label{eq:H2}
    nH_2\left(\frac{\onenorm{\psi_R-\Gamma_R}}{2}\right) < \log M,
\end{equation}
then $\cc(\psi)> \log d$.
\end{prop}
\begin{proof}
The proof goes by contradiction.  Suppose we have a pair of $\{\psi,\Gamma\}$ and some $d$ such that the  condition (\ref{eq:H2})  is satisfied, but $\cc(\psi)\leq \log d$. Namely, suppose there exists a depth-$t$ circuit $U$ such that
\begin{equation}\label{eq:U}
    \ket{\psi}=U\ket{0}^{\otimes n},
\end{equation}
where $t \leq\log d$, which guarantees that for any qubit $i$ the size of its light cone  $|L_i|\leq d$. %Suppose $\onenorm{\psi_{L_i}-\Gamma_{L_i}}\leq\delta$. Choosing an arbitrary state $\nu_{\ols{L_i}}$ on the rest of the system, we have%\ziwen{define notations, maybe just use $\bar{i}$}
%\begin{equation}
    %\onenorm{\tr_{\overline{ i}}(U^\dagger \psi_{L_i}\otimes \nu_{\ols{L_i}} U)-\tr_{\overline{ i}}(U^\dagger \Gamma_{L_i}\otimes \nu_{\ols{L_i}} U)}\leq \delta.
%\end{equation}
% by monotonocity.
It is clear that qubits outside the light cone $L_i$ are not correlated with $i$. As a result, assigning an arbitrary state $\nu_{\ols{L_i}}$ outside $L_i$ and running the circuit $U$ backwards should yield the same state on $i$, that is, 
\begin{align}
%\begin{split}
    &\tr_{\overline{ i}}(U^\dagger \psi_{L_i}\otimes \nu_{\ols{L_i}} U)=\tr_{\overline{ i}}(U^\dagger \psi U)=|0\rangle\langle0|, \label{eq:psi-nu}\\
    &\tr_{\overline{ i}}(U^\dagger \Gamma_{L_i}\otimes \nu_{\ols{L_i}} U)=\tr_{\overline{ i}}(U^\dagger \Gamma U),\label{eq:Theta-nu}
%\end{split}
\end{align}
where the first line follows from the definition (\ref{eq:U}).  Note that 
\begin{align}
   & \onenorm{\tr_{\overline{ i}}(U^\dagger \psi_{L_i}\otimes \nu_{\ols{L_i}} U)-\tr_{\overline{ i}}(U^\dagger \Gamma_{L_i}\otimes \nu_{\ols{L_i}} U)}\leq \onenorm{\psi_{L_i}-\Gamma_{L_i}}
\end{align}
 by monotonicity, so using (\ref{eq:psi-nu}) and (\ref{eq:Theta-nu}) we obtain
\begin{equation}
    \onenorm{\tr_{\overline{ i}}(U^\dagger \Gamma U)-|0\rangle\langle0|}\leq \onenorm{\psi_{L_i}-\Gamma_{L_i}}.
\end{equation}
Then using continuity of the von Neumann entropy $S$ we can  bound $S(\Gamma)$ as follows:
\begin{align}
    S(\Gamma)&=S(U^\dagger\Gamma U)\leq \sum_{i=1}^n S(\tr_{\overline{ i}}(U^\dagger \Gamma U))  \leq nH_2\left(\max_{L_i}\frac{\onenorm{\psi_{L_i}-\Gamma_{L_i}}}{2}\right), 
\end{align}
where the first line follows from the subadditivity of $S$ and the second line follows from the  Fannes--Audenaert inequality.
However, note that  $S(\Gamma)=\log m$ by definition and $|L_i|\leq d$, so the above is in contradiction with (\ref{eq:H2}).  Therefore, we must have $\cc(\psi)>\log d$.
\end{proof}
\begin{proof}[Proof of Theorem~\ref{thm:app_all_comp}]
In Proposition~\ref{lem:app_all_comp}, choose $|\psi\rangle$ to be a code state $|\psi\rangle\in\mathfrak{C}$ and $\Gamma$ to be the maximally mixed state of $\mathfrak{C}$. By the definition of subsystem invariance, for any set of qubits $R$ with $|R|\leq d$,
\begin{equation}
    \onenorm{\psi_R-\Gamma_R}\leq \varepsilon_\mathsf{g}(\mathfrak{C},d).
\end{equation}
For any state $\ket{\psi'}$ within the $\delta$-vicinity of a code state $\ket{\psi}$, namely $\onenorm{\psi'-\psi}\leq \delta$, we have $\onenorm{\psi'_R-\psi_R}\leq \delta$ by monotonicity. So we have
\begin{equation}
    \onenorm{\psi_R'-\Gamma_R}\leq \varepsilon_\mathsf{g}(\mathfrak{C},d)+\delta,
\end{equation}
by triangle inequality.
According to (\ref{eq:app_error_variance}), the above condition is automatically satisfied if
\begin{equation}
    \onenorm{\psi_R'-\Gamma_R}\leq 4\max_{\mathcal{N}}\tilde\varepsilon(\mathcal{N},\mathcal{E})+\delta,
\end{equation}
where $\mathcal{N}$ is any $d$-qubit replacement channel.
Then by applying Proposition~\ref{lem:app_all_comp} we obtain the statements in the theorem.
\end{proof}

On the other hand,  it is worth noting that 
\begin{equation}
    \max_{\psi\in\mathfrak{C}}\onenorm{\psi_R-\Gamma_R} \geq  2^{-k}\max_{\mathcal{N}}\tilde\varepsilon(\mathcal{N},\mathcal{E})^2.
\end{equation}
As a result,  
\begin{equation}
    k\leq n H_2(2^{-k-1}\max_{\mathcal{N}}\tilde\varepsilon(\mathcal{N},\mathcal{E})^2), 
\end{equation}
implies that there always exists a state in the code subspace that makes the subsystem variance too large to induce any circuit complexity bound.

We now move on to the case of  finite-dimensional spaces where there are geometric locality constraints. Note that Theorem~\ref{thm:geo_comp} in the main text is a simplified form of the result obtained by a similar argument as in Proposition~\ref{lem:app_all_comp}.  We now prove the following refined form:
\begin{thm}
\label{thm:app_geo_comp}
Given an $(\!(n,k)\!)$ code $\mathfrak{C}$, for any code state $\ket{\psi}\in\mathfrak{C}$, the $\delta$-robust geometric circuit complexity with respect to adjacency graph ${\mathsf{G}}_D$ embedded in a $D$-dimensional integer lattice satisfies 
\begin{equation}
    \cc^\delta_{{\mathsf{G}}_D}(\psi)> \frac{1}{2} d^{1/D}\left(1-(\Tilde{d}/d)^{1/D}\right),
\end{equation} 
if
\begin{align}
    k>&~ C(\mathsf{G}_D,\tilde d)H_2\left(\frac{1}{2}\varepsilon_{\mathsf{G}_D}(\mathfrak{C},d) + \frac{\delta}{2}\right)+C(\mathsf{G}_D,\tilde d)\log(2^{\tilde d}-1)\left(\frac{1}{2}\varepsilon_{\mathsf{G}_D}(\mathfrak{C},d) + \frac{\delta}{2}\right),
\end{align}
or
\begin{align}
    k>&~ C(\mathsf{G}_D,\tilde d)H_2\left(2\max_{\mathcal{N}}\tilde\varepsilon(\mathcal{N},\mathcal{E}) + \frac{\delta}{2}\right)+C(\mathsf{G}_D,\tilde d)\log(2^{\tilde d}-1)\left(2\max_{\mathcal{N}}\tilde\varepsilon(\mathcal{N},\mathcal{E}) + \frac{\delta}{2}\right),
\end{align}
where $C(\mathsf{G}_D,\Tilde{d})$ is the minimum number of size $\Tilde{d}$ hypercubes that can cover  $\Lambda$, $\mathcal{N}$ is any replacement channel acting on any $d$-qubit local region that is connected with respect to ${\mathsf{G}}_D$,  and the arguments of $H_2$ are assumed to be less than $1/2$.
\end{thm}

Again, we first introduce a general result Proposition~\ref{lem:app_geo_comp},  and then apply it to code spaces and take robustness into account to obtain the theorem statements.  The result is independently useful for proving circuit complexity bounds for specific states.  An interesting use case will be given in the discussion of momentum codes.

    \begin{prop}
   \label{lem:app_geo_comp}
   Let $\psi = \ketbra{\psi}{\psi}$ be a $n$-qubit state and let
\begin{equation}
\Gamma \coloneqq \frac{1}{M} \sum_{i=1}^{M}|\phi_i\rangle\langle\phi_i|
\end{equation}
be a reference state where the $\ket{\phi_i}$'s are orthogonal to each other (namely, $\Gamma$ is a maximally mixed state of the space  spanned by basis $\{\ket{\phi_i}\}$). Consider an associated  adjacency graph $\mathsf{G}_D$ embedded in a $D$-dimensional integer lattice.
If there exists some $\Tilde{d} = cd, c\in (0,1)$  such that for any local region $R$ connected with respect to ${\mathsf{G}}_D$ with $|R|\leq d$, it holds that 
\begin{align}
\label{eq:geo_H2}
    \log M>&~ C(\mathsf{G}_D,\Tilde{d})H_2\left(\frac{\onenorm{\psi_R-\Gamma_R}}{2}\right)+C(\mathsf{G}_D,\tilde d)\log(2^{\tilde d}-1)\frac{\onenorm{\psi_R-\Gamma_R} }{2},
\end{align}
where $C(\mathsf{G}_D,\Tilde{d})$ is the minimum number of size $\Tilde{d}$ hypercubes that can cover  $\Lambda$, 
then we have
\begin{equation}
    \cc_{\mathsf{G}_D}(\psi)> \frac{1}{2} d^{1/D}\left(1-(\Tilde{d}/d)^{1/D}\right).
\end{equation}
%$\|\psi_R-\Gamma_R\|\leq \delta$ which satisfies $\log m\geq nH_2(\delta/2) $ \ziwen{So essentially the bound is a function of $m$: If on any local region $R$ with $|R|<d$ it holds that $\|\rho_R-\Gamma_R\| < \tilde{O}\left((\frac{\log m}{n})^{\beta}\right), \beta > 1$, $\tilde{O}$ is the soft big-O notation that includes polylog factors}, then
%\begin{equation}
    %\cc(\psi)\geq \log_\gamma d.
%\end{equation}
\end{prop}
\begin{proof}
The intuition is reminiscent of the all-to-all case. However, in the all-to-all case, we divide the system into single qubits and bound their von Neumann entropy. For the current geometric case, we can instead divide the system into hypercubes $B_i$ of size $cd$. This is not possible for cases without geometric locality because in that case $|L_{B_i}|>d$ for any depth-$t$ circuit with $t=\omega(1)$.

Suppose we have a pair of $\{\psi,\Gamma\}$ defined on a $D$-dimensional lattice and some $d$ such that the  condition (\ref{eq:geo_H2}) is satisfied, but 
\begin{equation}
    \cc_{\mathsf{G}_D}(\psi)\leq\frac{1}{2} d^{1/D}\left(1-(\Tilde{d}/d)^{1/D}\right).
\end{equation}
Namely, suppose there exists a $D$-dimensional circuit $U_D$ with depth $t\leq d^{1/D}(1-(\Tilde{d}/d)^{1/D})/2$ such that $\ket{\psi}=U_D \ket{0}^{\otimes n}$.  This guarantees that for the covering $\{B_{i}\}$ of the lattice where $i=1,\ldots,C(\mathsf{G}_D,cd)$, we can divide the lattice into disjoint subsets $\{C_{i}\in B_{i}\}$ such that $|L_{C_i}|\leq d$. Note that from the division, $C(\mathsf G_D,\tilde d)\approx n/\tilde d$. Then similar to the proof for the  all-to-all case,
\begin{align}
    S(\Gamma)=&~S(U_D^\dagger\Gamma U_D)\\ \leq&~ \sum_{C_i}S(\tr_{\overline{ C_i}}(U_D^\dagger \Gamma U_D)) \\ 
    \leq&~ C(\mathsf{G}_D,\Tilde{d})H_2\left(\max_{C_i}\frac{\onenorm{\psi_{L_{C_i}}-\Gamma_{L_{C_i}}}}{2}\right)+\sum_{C_i}\frac{\onenorm{\psi_{L_i}-\Gamma_{L_i}}}{2}\log(2^{|C_i|}-1) \\
    \leq&~ C(\mathsf{G}_D,\tilde d)H_2\left(\max_{|R|\leq d}\frac{\onenorm{\psi_R-\Gamma_R}}{2}\right)+C(\mathsf{G}_D,\tilde d)\log(2^{\tilde d}-1)\max_{|R|\leq d}\frac{\onenorm{\psi_R-\Gamma_R} }{2}. 
\end{align}
This is in contradiction with (\ref{eq:geo_H2}). Therefore, we must have
\begin{equation}
    \cc_{\mathsf{G}_D}(\psi)\geq \frac{1}{2} d^{1/D}\left(1-(\Tilde{d}/d)^{1/D}\right).
\end{equation}
\end{proof}
\begin{proof}[Proof of Theorem~\ref{thm:app_geo_comp}]
In Proposition~\ref{lem:app_geo_comp}, choose $|\psi\rangle$ to be a code state $|\psi\rangle\in\mathfrak{C}$ and $\Gamma$ to be the maximally mixed state of $\mathfrak{C}$. Similar to the proof of Theorem~\ref{thm:app_all_comp}, using triangle inequality we obtain that for any state $\ket{\psi'}$ within the $\delta$-vicinity of a code state $\ket{\psi}$, namely $\onenorm{\psi'-\psi}\leq \delta$,
\begin{equation}
    \onenorm{\psi_R'-\Gamma_R}\leq \varepsilon_{\mathsf{G}_D}(\mathfrak{C},d)+\delta.
\end{equation}
According to (\ref{eq:app_error_variance}), the above condition is automatically satisfied if
\begin{equation}
    \onenorm{\psi_R'-\Gamma_R}\leq 4\max_{\mathcal{N}}\tilde\varepsilon(\mathcal{N},\mathcal{E})+\delta,
\end{equation}
where $\mathcal{N}$ is any $d$-qubit replacement channel.
Then by applying Proposition~\ref{lem:app_geo_comp} we obtain the statements in Theorem~\ref{thm:app_geo_comp}.
\end{proof}

In the above discussion, we have provided detailed derivations for the prototypical cases of the all-to-all graph and regular lattices.
Our approach can be easily generalized to obtain the following results for arbitrary adjacency graphs:
\begin{thm}
\label{thm:app_arbitrary}
Given an $(\!(n,k)\!)$ code $\mathfrak{C}$, for any code state $\ket{\psi}\in\mathfrak{C}$, the $\delta$-robust circuit complexity with respect to adjacency graph $\mathsf{G}$  satisfies 
\begin{equation}
    \cc^{\delta}_{\mathsf{G}}(\psi)>f^{-1}(d),
\end{equation}
where $f(t)$ is the maximal size of the light cone of a single qubit under depth-$t$ circuits, if
\begin{equation}
    H_2(\varepsilon_\mathsf{G}(\mathfrak{C},d)/2 + \delta/2) < k/n,
\end{equation}  or \begin{equation}
    H_2(2\max_{\mathcal{N}}\tilde\varepsilon(\mathcal{N},\mathcal{E}) + \delta/2) < k/n,
\end{equation}
where $\mathcal{N}$ is any replacement channel acting on any $d$-qubit subsystem that is connected with respect to ${\mathsf{G}}$,  and the arguments of $H_2$ are assumed to be less than $1/2$.
\end{thm}
\begin{prop}\label{lem:app_arbitrary}
   Let $\psi = \ketbra{\psi}{\psi}$ be a $n$-qubit state and let
\begin{equation}
\Gamma \coloneqq \frac{1}{M} \sum_{i=1}^{M}|\phi_i\rangle\langle\phi_i|
\end{equation}
be a reference state where the $\ket{\phi_i}$'s are $n$-qubit states orthogonal to each other (namely, $\Gamma$ is a maximally mixed state of the space  spanned by basis $\{\ket{\phi_i}\}$).  
If for any $d$-qubit subsystem that is connected with respect to ${\mathsf{G}}$, it holds that 
\begin{equation}\label{eq:H2}
    nH_2\left(\frac{\onenorm{\psi_R-\Gamma_R}}{2}\right) < \log M,
\end{equation}
then $\cc(\psi)> f^{-1}(d)$ where $f(t)$ is the maximal size of the light cone of a single qubit under depth-$t$ circuits with respect to $\mathsf{G}$.
\end{prop}
These can be proven by directly adapting the proofs of Theorem~\ref{thm:app_all_comp} and Proposition~\ref{lem:app_all_comp}. The bounds here are expressed in the most general forms.
For graphs with specific structures such as those embedded in more general Euclidean lattices or other physical geometries, we expect refinements like in Theorem~\ref{thm:app_geo_comp} to be feasible.

%\todo{This holds for all code states, namely $\varepsilon$ is for worst error location.  So this directly applies to unknown (random) error location. }

%Explain the intermediate regime.  If we take out code states (reduce $k$) like for momentum codes what 

As remarked in the main text, a converse implication of the circuit complexity theorems that may be useful is that the lowest complexity of code states induces lower bounds on the subsystem variance and AQEC imprecision.  In other words, any code that contains some state with a certain low circuit complexity is subject to a corresponding precision limit.  %A more formal statement is as follows.
\begin{cor}
In Theorems~\ref{thm:all_comp}, \ref{thm:geo_comp}, \ref{thm:app_all_comp}, \ref{thm:app_geo_comp}, \ref{thm:app_arbitrary}, for any $d$, if there is a state $\ket\psi\in\mathfrak{C}$ in the code space such that the circuit complexity lower bound is not satisfied, then the corresponding conditions for $\varepsilon$ and $\tilde\varepsilon$ cannot be satisfied.
\end{cor}

\section{Redundant encoding}
\label{app:redundant}
    Here we consider the code properties of a trivial kind of ``encoding'' map based on simply appending garbage states to the original logical state, that is,  $\mathcal{E}(\ket{\psi}) =  \ket{\psi}\otimes\ket{\text{gar}}\otimes\ldots\otimes\ket{\text{gar}}$ where $\bracket{\psi}{\text{gar}} = 0$, 
    which we call redundant encoding.  Of course, it is not an interesting  quantum code  in any meaningful sense.  However, we would like to make the conceptual point that such encoding can achieve a vanishingly small code error in a natural setting, meaning that a desirable notion of AQEC requires 
more than just the code error being vanishingly small.  

As a basic example, let $\ket{\psi}$ be a single qubit state and the noise $\mathcal{N}$ be a single erasure at a random unknown location. The redundant encoding naturally achieves %$\varepsilon = O(1/n)$ and
QEC inaccuracy $\tilde\varepsilon = O(1/\sqrt{n})$, the intuition behind which is that the location of the original state is affected with probability $1/n$, so on average at most $O(1/n)$ fraction of   the logical information  is leaked into the environment. 
This can be more precisely linked to the code error via the  complementary channel characterization,  Eq.~(\ref{eq:complementary-worst}).
More explicitly,  the (average) erased subsystem is  given by
\begin{equation}\label{eq:redundant-erased-1}
\rho = \frac{1}{n}\ket{\psi}\bra{\psi} + \left(1-\frac{1}{n}\right)\ket{\text{gar}}\bra{\text{gar}}.
\end{equation}
%It is straightforward that an average notion of subsystem variance obeys $\varepsilon = \frac{1}{n}\onenorm{\psi - \mathbbm{1}/2} = 1/n$. 
According to (\ref{eq:redundant-erased-1}), $\widehat{\mathcal{N} \circ \mathcal{E}}(\ket{\psi}\bra{\psi})=\frac{1}{n}\ket{\psi}\bra{\psi} + \left(1-\frac{1}{n}\right)\ket{\text{gar}}\bra{\text{gar}}$, so using $P(\rho, \ket{\text{gar}}\bra{\text{gar}}) =  \sqrt{1/n}$  it can be seen that $\tilde\varepsilon = \sqrt{1/n}$ from Eq.~(\ref{eq:complementary-worst}).

In our context, it is meaningful to consider a general setting where the logical state $\ket{\psi}$ has $k$ qubits and the  erasure error acts on $d$ unknown locations. Then the (average) erased subsystem has the form

\begin{equation}
    \rho=\left(1-\frac{{n-k \choose d}}{{n \choose d}}\right)\rho'+\frac{{n-k \choose d}}{{n \choose d}}|\text{gar}\rangle\langle \text{gar}|^{\otimes d},
\end{equation}
where $\rho'$ is a density matrix dependent on $\ket{\psi}$ associated with the situations where not only garbage locations are erased. 
In this case %the subsystem variance is bounded as 
%\begin{align}
 %\varepsilon  =   \left(1-\frac{{n-k \choose d}}{{n \choose d}}\right) \onenorm{\rho' - \mathbbm{1}/2^d} <2\left(1-\frac{{n-k \choose d}}{{n \choose d}}\right) \leq{2-2\left(1-\frac{k}{n-d+1}\right)^d}\leq{\frac{2dk}{n-d+1}},
%\end{align}
%which implies $\varepsilon = O(dk/n)$.
%As for $\tilde\varepsilon$, 
it is evident that $f(\rho',|\text{gar}\rangle\langle \text{gar}|^{\otimes d}) = 0$,   so
%$\rho$ can be fully expressed.
\begin{align}
   P(\rho, |\text{gar}\rangle\langle \text{gar}|^{\otimes d}) =&~\sqrt{1-\frac{{n-k \choose d}}{{n \choose d}}}\leq\sqrt{1-\left(1-\frac{k}{n-d+1}\right)^d}\leq\sqrt{\frac{dk}{n-d+1}},
\end{align}
which implies that the QEC inaccuracy satisfies $\tilde\varepsilon \leq O(\sqrt{dk/n})$.  The scalings are not significantly affected  when taking e.g.~logarithmic $d$ to draw fair comparisons with the demands of  superconstant complexity bounds.

%In a similar vein, we may consider an encoding map that generates a mixture or superposition of $\ket{\psi}$ at random locations, in which case we have $O(dk/n)$ subsystem variance.
The  message is that  such trivial encodings are already capable of generating ``AQEC codes'' that naturally achieves a  vanishingly small code error, and in fact, closely attaining the critical error scalings, which indicates that our theory provides meaningful, in some sense tight, criteria for ``interesting'' AQEC codes.

%Thus it is always in the trivial region.\todo{Matches the unknown error trivial scaling.}

%\jinmin{Can we show the other way around? i.e. get a lower bound for $\epsilon_{worst}$}

%\begin{widetext}

\section{Ferromagnetic Heisenberg chain code \label{app:heisenberg}}

Here we elaborate on the  properties of the AQEC code originating from the ferromagnetic Heisenberg spin chain model, refining the original analysis in Ref.~\cite{BCSB19:aqec} 
 along the way.  %Here the  lower bounds on the code error turns out to be particularly interesting whereas Ref.~\cite{BCSB19:aqec} concerns upper bounds. 
 While Ref.~\cite{BCSB19:aqec} focuses on upper bounds for the code error, our key finding here is that the subsystem variance of the code scales as $\tilde\omega(1/n)$ under  $d=\Theta(\log(n))$ and as a result, the code falls outside nontrivial complexity regime, which is consistent with the fact that the ground subspace contains trivial product states, e.g.,~$|0\rangle^{\otimes n}$.  This  provides a physical example of an AQEC code outside the nontrivial regime.

Specifically,  we consider the 1D ferromagnetic Heisenberg model, 
\begin{equation}
    H=-\frac{1}{2}\sum_{j=1}^n(\sigma^x_j\sigma^x_{j+1}+\sigma^y_j\sigma^y_{j+1}+\sigma^z_j\sigma^z_{j+1}),
\end{equation}
with periodic boundary conditions. Consider the $\sigma^z$-basis states $\ket{\mathbf{s}}$, where $\mathbf{s}=(s_1,s_2,...,s_n)\in\{1,-1\}^n$. Let $M=\sum_{i=j}^n \sigma_j^z$ be the total magnetization. This is a conserved quantity of the model, which gives an $n$-fold degenerate ground space, with the ground states $|h_m^n\rangle$ labelled by the magnetization 
\begin{equation}
    m={-n, -n+2, -n+4,\ldots,n-2, n},
\end{equation}
given by 
\begin{equation}
    \ket{h_m^n}=\frac{1}{\sqrt{n\choose n / 2+ m/ 2}}\sum_{\mathbf{s}:M(\mathbf{s})=m}\ket{\mathbf{s}}
\end{equation}
These states form a spin-$n/2$ representation of the $SO(3)$ symmetry. Note that there is a permutation symmetry for any $\ket{h_m^n}$, ensuring that the reduced density matrix for any $d$-qubit subsystem is the same. 
Ref.~\cite{BCSB19:aqec} considers code subspaces of the form
\begin{equation}
    \mathfrak{C}=\mathrm{Span}\left\{\ket{h_m^n}: m=-M,-M+2d+1,\ldots,M\right\},
\end{equation}
and shows that they yield AQEC codes with vanishingly small code error with respect to parameters $k=\Theta(\log (n))$, $d=\Theta(\log (n))$.
Here we would like to lower-bound the code error generally for $k,d=O(\log(n))$. The calculations largely follow those in Ref.~\cite{BCSB19:aqec}, with adjustments made when needed.
The reduced density matrix of $\ket{h_m^n}$ on a $d$-qubit subsystem is given by 
\begin{equation}
    \rho_d(m, n)=\sum_{r=-d}^d\left|h_r^d\right\rangle\left\langle h_r^d\right|\left[\frac{{d\choose d / 2+r / 2}{n-d \choose n / 2-d / 2+m / 2-r / 2}
}{{n \choose n / 2+m / 2}}\right].
\end{equation}
Since $\left|m-m^{\prime}\right| \leq 2M$ and $0 \ll d \ll 2M \ll n^{1 / 2}$,  the reduced density matrix can be well approximated as\footnote{Here an error in (D7) in Ref.~\cite{BCSB19:aqec} is corrected.}
\begin{equation}\label{eq:approx}
   \rho_d(m, n)\approx \sum_{r=-d}^d\left|h_r^d\right\rangle\left\langle h_r^d\right|
   {d\choose d / 2+r / 2}
   \frac{2^{-d} \sqrt{n}}{\sqrt{n-d}}\exp \left[\frac{1}{2}\left(\frac{m^2}{n}-\frac{(m-r)^2}{n-d}\right)\right],
\end{equation}
using the asymptotic approximation of the binomial coefficients   ${a \choose a / 2+b / 2}\approx \frac{2^{a+1}}{\sqrt{2 \pi a}} e^{-b^2 / 2 a}$. {Note that these approximations do not affect the scaling results.}
Given the reduced reference state
\begin{equation}
    \Gamma_d(n)=\frac{1}{2^k}\sum_m \rho_d(m, n),
\end{equation}
from (\ref{eq:approx}) we obtain 
\begin{align}
    \onenorm{\rho_d(m, n)-\Gamma_d(n)}\approx\sum_{r=-d}^d & {d\choose d / 2+r / 2}
   \frac{2^{-d} \sqrt{n}}{\sqrt{n-d}}\nonumber\\
   &\times\left|\exp \left[\frac{1}{2}\left(\frac{m^2}{n}-\frac{(m-r)^2}{n-d}\right)\right]-\frac{1}{2^k}\sum_{m'}\exp \left[\frac{1}{2}\left(\frac{m'^2}{n}-\frac{(m'-r)^2}{n-d}\right)\right]\right|.
\end{align}
Taylor expanding the term inside the absolute value, we obtain
\begin{equation}
    \left|\frac{m^2-m^{\prime 2}}{n}-\frac{(m-r)^2-\left(m^{\prime}-r\right)^2}{n-d}\right|  =\left|\frac{2\left(m-m^{\prime}\right)r}{n-d}-\frac{d\left(m^2-m^{\prime 2}\right)}{n(n-d)}\right|,
\end{equation}
For non-zero $r$ and $m\neq m'$, the leading order contribution 
 will be the first term, thus we obtain (the ``$\sim$'' symbol denotes the leading order approximation)
\begin{align}
    \onenorm{\rho_d(m, n)-\Gamma_d(n)}&\sim\sum_{r=-d,r\neq 0}^d {d\choose d / 2+r / 2}
   \frac{2^{-d} \sqrt{n}}{\sqrt{n-d}}
   \left|\frac{1}{2^k}\sum_{m'}\left[\frac{1}{2}\left( \frac{2\left(m-m^{\prime}\right)r}{n-d}-\frac{d\left(m^2-m^{\prime 2}\right)}{n(n-d)}\right)\right]\right|\\
   &\sim\sum_{r=-d}^d {d\choose d / 2+r / 2}
   \frac{2^{-d} \sqrt{n}}{\sqrt{n-d}}
   \left|\frac{mr}{n-d}\right|\\
   &\gtrsim\sum_{r=-d}^d {d\choose d / 2+r / 2}
   \frac{2^{-d}\sqrt{n}|m|}{({n-d})^{3/2}}-{d\choose d / 2}
   \frac{2^{-d}\sqrt{n}|m|}{({n-d})^{3/2}}\\
   &\sim\frac{\sqrt{n}|m|}{({n-d})^{3/2}}
   .
\end{align}
We now consider  superpositions of the code states in order to fully understand the code space property. For a general superposition state
\begin{equation}
    \ket{\psi}=\sum_m\alpha_m\ket{h_m^n},
\end{equation}
the reduced density matrix of a $d$-qubit subsystem $R$ is given by
\begin{equation}
    \rho_d(\psi,m)=\sum_{m,m'}\alpha_m\alpha^*_{m'}\tr_{\bar{R}}\ket{h_m^n}\bra{h_{m'}^n}.
\end{equation}
Note that for any $d$-qubit operator $\hat{O}$,
\begin{equation}
    \bra{h^n_{m'}}\hat{O}\ket{h^n_m}=0
\end{equation}
for $m\neq m'$ with $|m-m'|>2d$, because any $d$-qubit operator can at most flip the total magnetization $m$ by $2d$. Therefore we have
\begin{equation}
\tr_{\bar{R}}\ket{h_m^n}\bra{h_{m'}^n}=\delta_{m,m'}\rho_d(m, n),
\end{equation}
so the off-diagonal terms do not contribute to the subsystem variance. Then by convexity, the subsystem variance is given by 
\begin{equation}    \varepsilon_\mathsf{G}(\mathfrak{C},d)=\max_{m}\onenorm{\rho_d(m, n)-\Gamma_d(n)},
\end{equation}
where $\mathsf{G}$ can be any graph.
Since $\max |m|= M\sim d\cdot 2^k$, suppose we take $k = a\log(n)$, $d = b\log(n)$, the subsystem variance obeys the scaling 
\begin{equation}
    \varepsilon_\mathsf{G}(\mathfrak{C},d)=\omega(1/n^{1-a}).\label{eq:heisenberg_scaling}
\end{equation}
 So we observe that for sufficiently small $k$, in particular $k=o(\log(n))$,  considering that the Heisenberg chain code can contain trivial states, (\ref{eq:heisenberg_scaling}) indicates that the subsystem variance conditions in our complexity theorems cannot be significantly improved.

Note that, regarding the upper bounds, it can be checked that the bug in the analysis of Ref.~\cite{BCSB19:aqec} does not affect its key conclusion that the code error is vanishingly small.

%\todo{comment on upper bound}
%\end{widetext}

\section{Momentum codes and code fragmentation \label{app:momentum}}

This section provides a self-contained exposition of  the notion of \emph{momentum codes}, which  encode the lattice momentum of quantum states as logical information.  Our construction is inspired by the  long-range entanglement property of quantum systems with translational symmetry and nontrivial lattice momenta, which is of significant physical interest from the perspective of  Lieb--Schultz--Mattis-type theorems, topological order, and so on   \cite{GioiaWang}.   
As we will explain in detail, the momentum codes provide physically motivated constructions that illustrate various interesting phenomena that may carry deep physical significance. For example, by truncating parts of the code (so $k$ becomes smaller), one can reduce the code error  and in turn establish new circuit complexity bounds. In fact, this can happen for {all} fragments upon fragmentation of a large code into smaller pieces.  Furthermore, certain momentum codes naturally exhibit ``marginal order'' with code error near the critical scaling, potentially shedding new light on the understanding of quantum order. 
This section can also be viewed as a particularly interesting application of our new method to prove circuit lower bounds.

%Inspired by Ref.~\cite{GioiaWang}, We demonstrate a family of states with physical significance that can constitute AQEC codes that sit near the nontriviality boundary.

The most basic form of momentum codes with a single quasiparticle excitation is defined as follows.
Consider the following generalizations of $W$-states associated with a 1D spin chain,
\begin{align}
    |W_p\rangle &=\frac{1}{\sqrt{n}}\sum_{x=1}^n e^{ipx}|\tilde x\rangle\\&=\frac{1}{\sqrt{n}}(e^{ip}\ket{100\ldots 0} + e^{i2p}\ket{010\ldots 0} + \ldots + e^{inp}\ket{000\ldots 1}),\quad p=\frac{2\pi m}{n},\quad m=0,1,\ldots,n-1.
\end{align}
where $\ket{\tilde{x}}$ corresponds to a quasiparticle excitation $\ket{1}$ at the $x$-th site and $\ket{0}$'s elsewhere, as elucidated by the second line.
$\ket{W_p}$ can be thought of as a generalized $W$-states that carries lattice momentum $p$, in the sense that it attains a $e^{ip}$ phase under the  lattice translation operator $\hat{T}$, i.e.,~$\hat{T}|W_p\rangle=e^{ip}|W_p\rangle$. 
%Here $|x\rangle=|00\ldots010\ldots0\rangle$ where we have $|0\rangle$ states on all sites except for $\ket{1}$  on the $x$-th location representing  a quasiparticle excitation.
The momentum can only take quantized values because of the periodic boundary condition. 

We start by considering the code space spanned by the entire set of $|W_p\rangle$ with all possible different momenta, namely, 
\begin{equation}
    \mathfrak{C} = \mathrm{Span}\left\{\ket{W_p}: p =\frac{2\pi m}{n}, m = 0,1,\ldots,n-1\right\},
\end{equation}
which we refer to as a full momentum code.  This code has $k=\log(n)$.
We now analyze its subsystem variance. For the reference state
\begin{equation}
    \Gamma=\frac{1}{n}\sum_p |W_p\rangle\langle W_p|,
\end{equation}
the reduced state of a $d$-local region (we simply consider the first $d$ qubits labeled by $A_d$, without loss of generality,  because of the periodic boundary condition) is given by:
\begin{equation}
    \Tr_{\ols{A_d}}\Gamma=\frac{1}{n}\left(\sum_{x=1}^d |\tilde x\rangle\langle \tilde x|+(n-d)|\tilde 0\rangle\langle\tilde 0|\right),
\end{equation}
where we have slightly stretched the notations, using $|\tilde x\rangle$  to denote the $d$-qubit state with excitation $|1\rangle$ at the $x$-th site and using  $|\tilde 0\rangle$ to denote the all-0 state
 $|0\rangle^{\otimes d}$.
%Recall that our main theorems can still imply circuit complexity bounds for code states that are sufficiently close to the reference state. 
For the momentum eigenstates
$|W_p\rangle$, The reduced states are given by
\begin{equation}  \Tr_{\ols{A_d}}|W_p\rangle\langle W_p|=\frac{1}{n}\left(\sum_{x,y=1}^d e^{ip(x-y)}|\tilde x\rangle\langle \tilde y|+(n-d)|\tilde 0\rangle\langle\tilde 0|\right),
\end{equation}
so we have
\begin{align}
   \onenorm{\Tr_{\bar{d}}|W_p\rangle\langle W_p|-\Tr_{\bar{d}}\Gamma}
    =&~\frac{1}{n}\onenorm{\sum_{x,y=1}^d e^{ip(x-y)}|\tilde x\rangle\langle \tilde y|-\sum_{x=1}^d |\tilde x\rangle\langle \tilde x|}\\
    =&~\frac{1}{n}\onenorm{d|\phi_{p}\rangle\langle\phi_{p}|-\sum_{p'} |\phi_{p'}\rangle\langle\phi_{p'}|}\\
    =&~\frac{2}{n}(d-1),
\end{align}
where $|\phi_{p}\rangle=\frac{1}{\sqrt{d}}\sum_{x=1}^de^{ipx}|\tilde x\rangle$. It is easy to verify the first equality does not depend on the choice of the code state. 
%Using lemmas in App.~\ref{app:aqec}, we can also obtain the bound
%\begin{equation}
%$\tilde\varepsilon \geq ({d-1})/{2n}$
%\end{equation}
 %under $d$-local complete noise.  
According to Proposition~\ref{lem:app_geo_comp},  the condition for our circuit complexity bound requires
\begin{equation}
    \log n> d-1+\frac{n}{d-1 }H_2\left(\frac{d-1}{n}\right),
\end{equation}
where we used $\tilde{d}<d-1$. 
Here both sides are $\log n$ at the first order, but up to the second order, the right-hand side is $\log n+d$.  That is, the condition is only weakly violated but it cannot be fulfilled by tuning $d$. %although $\cc_\mathsf{G}(W_p)=\omega(1)$ as is demonstrated in Ref.~\cite{GioiaWang}, 
Furthermore, since by suitable superpositions of $|W_p\rangle$ we can obtain   product states $\ket{\tilde x}$ in the code space, the subsystem variance obeys
\begin{align}
   \varepsilon_{{\mathsf{G}}_1}(\mathfrak{C},d)  \geq \onenorm{\Tr_{\ols{A_d}}|\tilde1\rangle\langle \tilde1|-\Tr_{\ols{A_d}}\Gamma}
    =2-\frac{2}{n},
\end{align}
which is nonvanishing.  By Proposition~\ref{thm:two-way}, the QEC inaccuracy  for replacement channels is also nonvanishing.
To summarize, the full momentum code is highly ``diverse'', which leads to the somewhat peculiar feature that even though the subsystem variance for an entire set of basis states is vanishingly small, the subsystem variance for the code space is not, due to the versatility of superpositions. This diverse nature of the full momentum code renders its properties  insufficient to induce nontrivial circuit complexity bounds.

%However, since $\cc_\mathsf{G}(W_p)=\omega(1)$ it will be ideal if there exists some AQEC code in the non-trivial region with $|W_p\rangle$ as its code state and our results apply. In fact, there is flexibility in code design to make our bounds work. 

%\subsection*{{Code fragmentation}}

Let us now consider the fragmentation of the full momentum code into smaller pieces. We first show that it is possible for all fragments of the code to exhibit vanishingly small subsystem variance under polynomial $d$, which is in contrast to the behavior of the full momentum code.  For instance, consider the extreme fragmentation into $k = 1$ codes, each spanned by a pair of
nearest-neighbor momentum eigenstates,
\begin{equation}
    \hat{\mathfrak{C}}_p = \mathrm{Span}\left\{\ket{W_p},\ket{W_{p+\Delta p}}\right\}, \quad\Delta p=2\pi/n,
\end{equation}
where $\Delta p=2\pi/n$ is the smallest momentum interval   determined the periodic boundary condition. 
For this code the reference state is
\begin{equation}
    \hat\Gamma_p=\frac{1}{2}(|W_p\rangle\langle W_p|+|W_{p+\Delta p}\rangle\langle W_{p+\Delta p}|).
\end{equation}
Here for the momentum eigenstate $\ket{W_p}$, 
simple calculations yield 
\begin{align}
\Tr_{\ols{A_d}}|W_p\rangle\langle W_p|-\Tr_{\ols{A_d}}\hat\Gamma_p=\frac{1}{2n}\sum_{x,y=1}^d \left(e^{i(x-y)p}-e^{i(x-y)(p+\Delta p)}\right)|\tilde x\rangle\langle \tilde y|,
\end{align}
so we have 
\begin{align}
        %\varepsilon(\hat{\mathfrak{C}}_p,d,{\mathsf{G}}_1)  =  
\onenorm{\Tr_{\ols{A_d}}|W_p\rangle\langle W_p|-\Tr_{\ols{A_d}}\hat\Gamma_p}&\leq\frac{1}{2n}\sum_{x,y=1}^d|e^{i(x-y)p}-e^{i(x-y)(p+\Delta p)}|\onenorm{|\tilde x\rangle\langle \tilde y|}\\
        &=\frac{1}{2n}\sum_{x,y=1}^d|e^{i(x-y)p}-e^{i(x-y)(p+\Delta p)}|\\
        &\leq\frac{1}{2n}\sum_{x,y=1}^d\left|2\sin\left(\frac{x-y}{2}\Delta p\right)\right|\\
        &\leq\frac{1}{2n}d^3\Delta p\\&=\frac{\pi d^3}{n^2}.
\end{align}
For $\ket{W_{p+\Delta p}}$ we have the same result.
In this case, according to Proposition~\ref{lem:app_geo_comp}, the condition for our circuit complexity bound is satisfied when
\begin{equation}
    1\geq \frac{\pi d^3}{2n}+\frac{n}{\tilde{d}}H_2\left(\frac{\pi d^3}{2n^2}\right),
\end{equation}
which is always the case when we take $d=O(n^a)$ with $a<1/3$. Therefore, by applying Proposition~\ref{lem:app_geo_comp} to all the fragmented codes we  prove that \begin{equation}
    \cc_{{\mathsf{G}}_1}(w)\geq \frac{1}{2} d > O(n^{\frac{1}{3} - \epsilon}), \quad\forall\epsilon>0,
\end{equation} 
where $\ket{w}$ can be any momentum eigenstate $\ket{W_p}$.  For the purpose of proving circuit complexity bounds for specific states, a key implication from the comparison between the full and fragmented momentum codes is that modifying the code that the states belong to can be useful.

To understand the code space properties, we also need to consider coherent superpositions of momentum eigenstates which takes the the form $|\psi_p\rangle=\alpha|W_p\rangle+\beta\ket{ W_{p+\Delta p}}$ where $\alpha,\beta$ are amplitudes. By definition, the subsystem variance of $\hat{\mathfrak{C}}_p$ is given by
\begin{equation}
    \varepsilon_{{\mathsf{G}}_1}(\hat{\mathfrak{C}}_p,d)  =\max_{\alpha,\beta}  \onenorm{\Tr_{\ols{A_d}}|\psi_p\rangle\langle \psi_p|-\Tr_{\ols{A_d}}\hat\Gamma_p}.
\end{equation}
The reduced states of $\ket{\psi_p}$ are calculated to be
\begin{align}
\Tr_{\ols{A_d}}|\psi_p\rangle\langle \psi_p|=&~\frac{1}{n}\sum_{x,y=1}^d \left(\alpha e^{ipx}+\beta e^{i(p+\Delta p)x} \right)\left(\alpha^*e^{-ipy}+\beta^*e^{-i(p+\Delta p)y}\right)|\tilde x\rangle\langle \tilde y|\nonumber\\
&+\frac{1}{n}\left((n-d)+\alpha\beta^*\frac{e^{-i\Delta p(d+1)}-e^{-i\Delta p }}{1-e^{-i\Delta p}}+\beta\alpha^*\frac{e^{i\Delta p(d+1)}-e^{i\Delta p }}{1-e^{i\Delta p}}\right)|\tilde 0\rangle\langle\tilde 0|.
\end{align}
Again, use the notation $\ket{\phi_p}=\frac{1}{\sqrt{d}}\sum_{x=1}^d e^{ipx}|\tilde x\rangle$. 
We have
\begin{align}
        \onenorm{\Tr_{\ols{A_d}}|\psi_p\rangle\langle \psi_p|-\Tr_{\ols{A_d}}\hat\Gamma_p}=&~\frac{1}{n}\onenorm{d(\alpha|\phi_p\rangle+\beta|\phi_{p+\Delta p}\rangle)(\alpha^*\langle\phi_p|+\beta^*\langle\phi_{p+\Delta p}|)-\frac{d}{2}|\phi_p\rangle\langle\phi_p|-\frac{d}{2}|\phi_{p+\Delta p}\rangle\langle\phi_{p+\Delta p}|}\nonumber\\
        &+\frac{1}{n}\onenorm{\left(\alpha\beta^*\frac{e^{-i\Delta p(d+1)}-e^{-i\Delta p }}{1-e^{-i\Delta p}}+\beta\alpha^*\frac{e^{i\Delta p(d+1)}-e^{i\Delta p }}{1-e^{i\Delta p}}\right)|\tilde0\rangle\langle\tilde0|}\\
        \leq&~\frac{d}{n}\left(\left|\alpha\alpha^*-\frac{1}{2}\right|+\left|\beta\beta^*-\frac{1}{2}\right|+|\alpha\beta^*|+|\beta\alpha^*|\right)+\frac{2}{n}\left|\alpha\beta^*\frac{\sin(\Delta p d/2)}{\sin(\Delta p/2)}\right|\\
        \leq&~ \frac{5d}{n},
\end{align}
so the scaling of the subsystem variance is bounded as 
\begin{equation}
    \varepsilon_{{\mathsf{G}}_1}(\hat{\mathfrak{C}}_p,d)=O\left(\frac{d}{n}\right),
\end{equation}
which is vanishingly small for any $d=o(n)$.

The above analysis and observations can be extended to momentum codes associated with more general quasiparticles, particularly in the natural case where the Fourier transform of the momentum eigenstates has excitations in a region of length $\xi$
that correspond to the size of the quasiparticle.
The phenomenon that the code fragmentation can improve the code error can also be observed. Explicitly, let us consider momentum eigenstates of the form 
\begin{equation}
|W_{p,\xi}\rangle=\frac{1}{\sqrt{n}}\sum_{x=1}^ne^{ipx}|\tilde{x}_\xi\rangle, \quad p=\frac{2\pi m}{n},\quad m=0,1,\ldots,n-1,  
\end{equation}
where $|\tilde{x}_\xi\rangle$ denotes a generalized version of $\ket{\tilde x}$ in which the size-$\xi$ region from the $x$-th site to the $(x+\xi-1)$-th site can host quasiparticle excitations ($\ket{1}$'s) and outside this region are $|0\rangle$'s. We again first consider the full momentum code spanned by the entire set of $|W_{p,\xi}\rangle$ with all possible different momenta, namely, 
\begin{equation}
    \mathfrak{C}_\xi = \mathrm{Span}\left\{\ket{W_{p,\xi}}: p =\frac{2\pi m}{n}, m = 0,1,\ldots,n-1\right\}.
\end{equation}
Again, this code has $k=\log n$ and  exhibits  nonvanishing code error.  We perform a similar calculation as in the single quasiparticle case.  Given the reference state
\begin{equation}
    \Gamma_\xi=\frac{1}{n}\sum_p |W_{p,\xi}\rangle\langle W_{p,\xi}|,
\end{equation}
its reduced state of a $d$-local region $A_d$ is given by
\begin{equation}
    \Tr_{\ols{A_d}}\Gamma_\xi=\frac{1}{n}\left(\sum_{x=2-\xi}^d \Tr_{\ols{A_d}}|\tilde x_\xi\rangle\langle \tilde x_\xi|+(n-d-\xi+1)|\tilde 0\rangle\langle\tilde 0|\right).
\end{equation}
Since by suitable superpositions of $|W_{p,\xi}\rangle$ we can obtain product states $\ket{\tilde x_\xi}$ in the code space, the subsystem variance obeys
\begin{align}
   \varepsilon_{{\mathsf{G}}_1}(\mathfrak{C}_\xi,d)  &\geq \onenorm{\Tr_{\ols{A_d}}|\tilde1_\xi\rangle\langle \tilde1_\xi|-\Tr_{\ols{A_d}}\Gamma_\xi}\\
   &\geq \left|\onenorm{|\tilde1_\xi\rangle\langle \tilde1_\xi|-\frac{n-d-\xi+1}{n}|\tilde 0\rangle\langle\tilde 0|}-\onenorm{\frac{1}{n}\sum_{x=2-\xi}^d \Tr_{\ols{A_d}}|\tilde x_\xi\rangle\langle \tilde x_\xi|}\right|\\
    &= 
    2-\frac{2(d+\xi-1)}{n}.
\end{align}
For $d+\xi<n$, the subsystem variance is   nonvanishing, and by Proposition~\ref{thm:two-way}, the QEC inaccuracy  for replacement channels is also nonvanishing.  We see again that the highly diverse nature of the full code leads to a large code error that is insufficient to induce nontrivial circuit complexity bounds.

Now consider the  fragmented $k=1$ codes given by
\begin{equation}
    \hat{\mathfrak{C}}_{p,\xi} = \mathrm{Span}\left\{\ket{W_{p,\xi}},\ket{W_{p+\Delta p,\xi}}\right\}, \quad\Delta p=2\pi/n,
\end{equation}
for which the associated reference state is
\begin{equation}
    \hat\Gamma_{p,\xi}=\frac{1}{2}(|W_{p,\xi}\rangle\langle W_{p,\xi}|+|W_{p+\Delta p,\xi}\rangle\langle W_{p+\Delta p,\xi}|).
\end{equation}
Following a similar calculation as before, we have   (note the periodic boundary condition)
\begin{align}
    %\begin{split}
    \onenorm{\Tr_{\ols{A_d}}|W_{p,\xi}\rangle\langle W_{p,\xi}|-\Tr_{\ols{A_d}}\hat\Gamma_{p,\xi}} &\leq\frac{1}{2n}\sum_{x,y=-\xi+2}^{d+\xi-1}|e^{i(x-y)p}-e^{i(x-y)(p+\Delta p)}|\onenorm{\tr_{\ols{A_d}}|\tilde{x}_\xi\rangle\langle \tilde{y}_\xi|}\\
        &\leq\frac{1}{2n}\sum_{x,y=-\xi+2}^{d+\xi-1}|e^{i(x-y)p}-e^{i(x-y)(p+\Delta p)}|\\
        &\leq\frac{1}{2n}\sum_{x,y=-\xi+2}^{d+\xi-1}\left|2\sin\left(\frac{x-y}{2}\Delta p\right)\right|\\
&<\frac{1}{2n}(d+2\xi-2)^3\Delta p\\&=\frac{\pi (d+2\xi-2)^3}{n^2},
\end{align}
where for the second line we used the fact that %\todo{double check}
$\tr_{\ols{A_d}}|\tilde{x}_\xi\rangle\langle \tilde{y}_\xi| = \delta_{xy}|\tilde 0\rangle\langle\tilde 0|$
except for $2-\xi\leq x,y\leq d+\xi-1$, and also  the monotonicity of the trace norm.
For $\ket{W_{p+\Delta p,\xi}}$ we have the same result.
When taking $d,\xi=O(n^a)$ with $a<1/3$, the circuit complexity bound condition is satisfied.  So by applying Proposition~\ref{lem:app_geo_comp} to all the
fragmented codes we prove that 
\begin{equation}
\cc_{{\mathsf{G}}_1}(w)\geq \frac{1}{2} d > O(n^{\frac{1}{3} - \epsilon}), \quad\forall\epsilon>0,
\end{equation}
where $\ket{w}$ can be any momentum eigenstate $\ket{\psi_{p,\xi}}$.

Given the coherent superpositions $|\psi_{p,\xi}\rangle=\alpha|W_{p,\xi}\rangle+\beta\ket{ W_{p+\Delta p,\xi}}$,
the subsystem variance has the form
\begin{equation}
    \varepsilon_{{\mathsf{G}}_1}(\hat{\mathfrak{C}}_p,d)  =\max_{\alpha,\beta}  \onenorm{\Tr_{\ols{A_d}}|\psi_{p,\xi}\rangle\langle \psi_{p,\xi}|-\Tr_{\ols{A_d}}\hat\Gamma_{p,\xi}}.
\end{equation}
The reduced states of $|\psi_{p,\xi}\rangle$ are calculated to be
\begin{align}
\Tr_{\ols{A_d}}|\psi_{p,\xi}\rangle\langle \psi_{p,\xi}|=&~\frac{1}{n}\sum_{x,y=-\xi+2}^{d+\xi-1} \left(\alpha e^{ipx}+\beta e^{i(p+\Delta p)x} \right)\left(\alpha^*e^{-ipy}+\beta^*e^{-i(p+\Delta p)y}\right)\Tr_{\ols{A_d}}|\tilde x\rangle\langle \tilde y|\nonumber\\
&+\frac{1}{n}\left((n-d-2\xi+2)+\alpha\beta^*\frac{e^{-i\Delta p(d+\xi)}-e^{-i\Delta p(2-\xi) }}{1-e^{-i\Delta p}}+\beta\alpha^*\frac{e^{i\Delta p(d+\xi)}-e^{i\Delta p (2-\xi)}}{1-e^{i\Delta p}}\right)|\tilde 0\rangle\langle\tilde 0|.
\end{align}
Let $\ket{\phi_{p,\xi}}=\frac{1}{\sqrt{d+2\xi-2}}\sum_{x=-\xi+2}^{d+\xi-1} e^{ipx}|\tilde x\rangle$.
We have
\begin{align}
        &\onenorm{\Tr_{\ols{A_d}}|\psi_{p,\xi}\rangle\langle \psi_{p,\xi}|-\Tr_{\ols{A_d}}\hat\Gamma_{p,\xi}}\nonumber\\
        =&~\frac{1}{n}\onenorm{(d+2\xi-2)\Tr_{\ols{A_d}}\left(\alpha|\phi_p\rangle+\beta|\phi_{p+\Delta p}\rangle\right)\left(\alpha^*\langle\phi_p|+\beta^*\langle\phi_{p+\Delta p}|\right)-\frac{d+2\xi-2}{2}\Tr_{\ols{A_d}}(|\phi_p\rangle\langle\phi_p|+|\phi_{p+\Delta p}\rangle\langle\phi_{p+\Delta p}|)}\nonumber\\
        &+\frac{1}{n}\onenorm{\left(\alpha\beta^*\frac{e^{-i\Delta p(d+\xi)}-e^{-i\Delta p(2-\xi) }}{1-e^{-i\Delta p}}+\beta\alpha^*\frac{e^{i\Delta p(d+\xi)}-e^{i\Delta p (2-\xi)}}{1-e^{i\Delta p}}\right)|\tilde0\rangle\langle\tilde0|}\\
        \leq&~\frac{d+2\xi-2}{n}\left(\left|\alpha\alpha^*-\frac{1}{2}\right|+\left|\beta\beta^*-\frac{1}{2}\right|+|\alpha\beta^*|+|\beta\alpha^*|\right)+\frac{2}{n}\left|\alpha\beta^*\frac{\sin(\Delta p (d+2\xi-2)/2)}{\sin(\Delta p/2)}\right|\\
        \leq&~\frac{5(d+2\xi-2)}{n},
\end{align}
so the scaling of the subsystem variance is bounded as 
\begin{equation}
    \varepsilon_{{\mathsf{G}}_1}(\hat{\mathfrak{C}}_p,d)=O\left(\frac{d+2\xi}{n}\right),
\end{equation}
which is vanishingly small for any $d,\xi=o(n)$.

We note that one can prove good circuit complexity bounds even for certain superpositions by calculating correlation functions and applying  Lieb--Robinson bounds, due to the geometric locality.  A point  we would like to emphasize here is that  the code properties of momentum codes, on their own, are insufficient to generate code space circuit complexity bounds. %In other words, it is possible to emulate certain behaviors of momentum codes  using trivial states.  
This is in stark contrast to the standard cases of gapped topological order.  
In this sense, the nontrivial order associated with such systems is of a ``marginal'' kind, which should be distinguished from the usual notion of topological order. 
This is of potential significance to the understanding and classification of quantum order, especially
considering the various intriguing connections between nontrivial momentum and quantum order \cite{GioiaWang}.

Of course, the above calculations are merely intended as a basic  demonstration of the properties of momentum codes and code fragmentation; all parameters  can be further generalized.
From the perspective of lower bounding the circuit complexity of certain states, the above examples indicate that our  method can transition from being inapplicable to applicable upon modifications of the code.  
To conclude, a potentially interesting perspective of momentum codes is that the code error can be used to understand the fundamental limitations on the  detectability of  momentum in quantum systems, which is of potential significance in the view of 
symmetries, uncertainty principles, etc. We leave further study of these aspects for future work.

\section{Topological entanglement entropy}\label{app:tee}
Here we provide detailed proofs for the topological entanglement entropy (TEE) results as well as an extended  discussion on the string net example. As is customary, we consider 2D systems defined on a torus.  We are most interested in states that obey an area law (that is a general property of e.g.~ground states of gapped local Hamiltonians), meaning that the entanglement entropy of a subsystem $A$ has the form  
\begin{equation}
S(A)=al(A)-b(A)\gamma + o(1).    
\end{equation}
The first term is the area law term where
 $a$ is a constant and $l(A)$ is the length of the boundary of $A$. The second term is the correction term where $b(A)$ is the number of connected components of the boundary of $A$ (for a contractible region, $b_0=1$) and $\gamma$ is the TEE.  
 
 TEE is widely used as a smoking gun for topological order. %Specifically, $b_0=1$ for a contractible region. 
 In line with our result which characterizes the robustness of nontrivial circuit complexity under relaxations of code properties, we would also like to understand to what extent the TEE signature is robust. This provides a more complete picture of the relationship between these widely used but different conditions for topological order. 
We restate our main results here  for readers' convenience.
\begin{prop}
   Consider an $(\!(n,k)\!)$ code defined on a 2D lattice on a torus.  Suppose that $\varepsilon=o(1/n)$ for any contractible region of size $d$, and there exists a code state with area-law entanglement. Then  in the thermodynamic limit the TEE of any code state satisfies
    \begin{equation}
        \gamma\geq k/\max\{2,2\lfloor n/2d\rfloor\}.
    \end{equation}
Specifically, we have the best bound $\gamma\geq k/2$, which is saturated by abelian topological order, if the code conditions hold for 
\begin{enumerate} [i)]
\item $d>n/4$, or
\item any $d$ linear in $n$ if, additionally, for error regions that do not contain non-contractible loops on the torus, $\tilde\varepsilon = o(1/n)$ can be achieved by recovery operations  acting within $O(1)$ distance to the error region ($\ell=O(1)$, $\ell$ parameter as defined in Ref.~\cite{Flammia2017limitsstorageof}).
\end{enumerate}

\end{prop}

\begin{proof}
The result is proven based on the  Markov entropy decomposition in Ref.~\cite{IssacKim-TOEntropy}, which enables us to relate the code parameters with TEE. 
Specifically, following Ref.~\cite{IssacKim-TOEntropy},
 we choose a sequence of subregions $A_i, B_i, C_i$ with $ i=1,2,\ldots,m$ with $m$ being the total number of steps, such that  $A_i B_i C_i=A_{i+1} B_{i+1}$,  $A_1 B_1$ and $B_i C_i$ are contractible regions with $|A_iB_i|,|B_iC_i|\leq d$ for some given $d$, and  $A_m B_m C_m$ is the entire torus. 
 As we require $|A_iB_i|,|B_iC_i|\leq d$,  the minimum number of steps needed for the sequence is $m=\max\{3, 2\lfloor n/2d\rfloor+1\}$. This is because we need an odd number ($\geq 3$) of steps for a torus geometry, and we can choose the subregions $B_i$ to be as thin as possible so we only need to require $(m+1)d>n$.
 Fig.~\ref{fig:markov-enytopy} illustrates a  viable sequence of subregion partitions with $m=3$.
 \begin{figure}
\includegraphics[width=0.6\columnwidth]{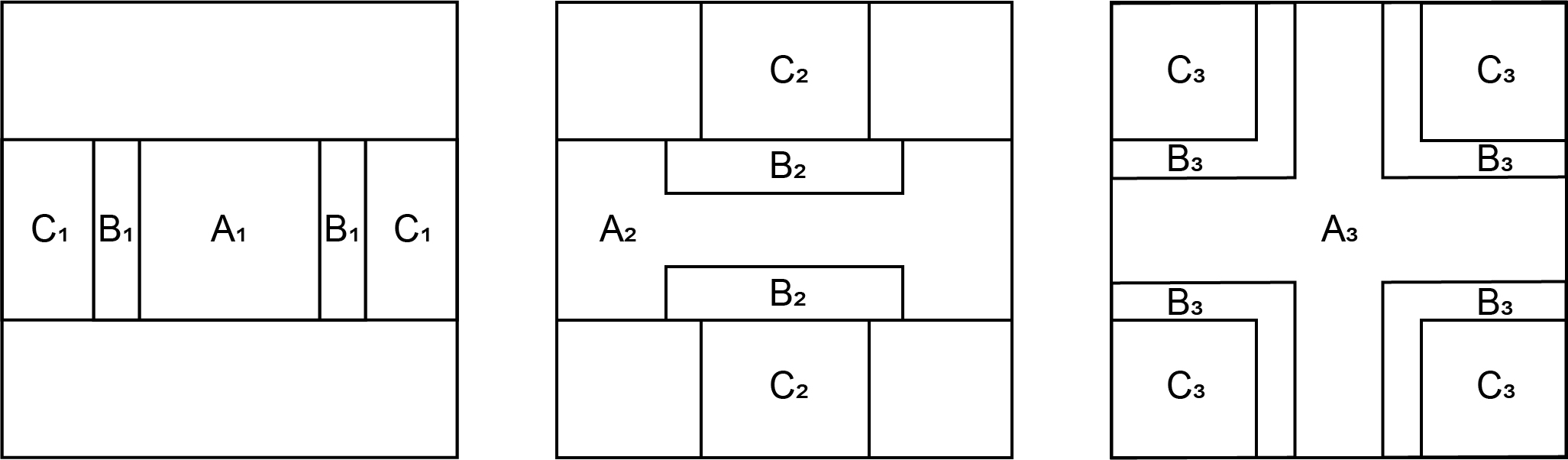}
\caption{(Adapted from Ref.~\cite{IssacKim-TOEntropy}) A viable sequence of subregion partitions with  $m=3$ for the Markov entropy decomposition. The system is defined on a torus.\label{fig:markov-enytopy}}
\end{figure}

The Markov entropy decomposition leads to the  inequality~\cite{IssacKim-TOEntropy}
\begin{equation}
    S(A_m B_m C_m)\leq S(A_1B_1)+\sum_{i=1}^{m} S(B_iC_i)-S(B_i),
    \label{eq:Markov}
\end{equation}
where $S(R)$  denotes the entanglement entropy of  subregion $R$.
We apply this inequality to the  maximally mixed code state (reference state) $\Gamma$. On the left-hand side, $S(A_m B_m C_m) = k$.  
To analyze the right-hand side, suppose the code has subsystem variance $\varepsilon$ on any contractible region of size $d$, i.e.,~for any code state $\psi$ and any contractible region $R$ with $|R|\leq d$ it holds that
$
   \onenorm{\psi_R-\Gamma_R}\leq\varepsilon
$.
Using the Fannes--Audenaert inequality, we obtain the following subregion entropy continuity bound,
\begin{equation}
    |S(\Gamma_R)-S(\psi_R)|\leq \frac{\varepsilon}{2}|R|+H_2(\varepsilon/2).
\end{equation}
By assumption, some code state $\phi$ has area law entanglement, namely $S(\phi_R) = al(R) - b(R)\gamma +o(1)$, which implies that 
\begin{equation}
S(\Gamma_R) = al(R) - b(R)\gamma + \lambda\left(\frac{\varepsilon}{2}|R|+H_2(\varepsilon/2)\right) +o(1),
\end{equation}
where $-1\leq\lambda\leq 1$, bounding the correction originating from subsystem variance.
When $\varepsilon=o(1/n)$, the correction vanishes in the thermodynamic limit, meaning that 
\begin{equation}
    S(\Gamma_R) = al(R) - b(R)\gamma  +o(1).
    \label{eq:Gamma_R}
\end{equation}
It can be observed that the combination of subregion entropies on the right-hand side is structured in a way that the contributions of the area law terms are canceled out.
Also, the partition can be designed in such a way that  $b(B_i) - b(B_i C_i) = 1$.
Inserting (\ref{eq:Gamma_R}) into (\ref{eq:Markov}), we obtain
\begin{equation}\label{eq:TEEbound}
    k \leq (m-1)\gamma + o(1),
\end{equation}
so in the thermodynamic limit we have
\begin{equation}
    \gamma\geq k/(m-1)=k/\max\{2,2\lfloor n/2d\rfloor\}.
\end{equation}

Specifically, the best lower bound $\gamma\geq k/2$ is attained when $d>n/4$ (in which case we can take $m=3$). This bound is already maximal as it is known to be achieved by abelian topological order.

The $d>n/4$ condition can be improved to any $d$ that is linear in $n$ if the following holds: for error regions that do not contain non-contractible loops on the torus, $\tilde\varepsilon = o(1/n)$ can be achieved by recovery operations acting within $O(1)$ distance to the error region. More explicitly, let us consider an additional code parameter $\ell$ that restricts the locality of recovery~\cite{Flammia2017limitsstorageof}: for a noise channel acting on region $A$, the recovery channel is only permitted to act within distance $\ell$ to $A$, namely in the $\ell$-vicinity of $A$. The improved result is derived using the Expansion Lemma \cite{Flammia2017limitsstorageof}. The error metric used here is slightly different but the proof can be directly adapted from that in Ref.~\cite{Flammia2017limitsstorageof}.
\begin{lem}[Expansion lemma \cite{Flammia2017limitsstorageof}]
    Let $A B=A^{+\ell}$ be the $\ell$-neighborhood of region $A$. If $A$ is $\left(\tilde\varepsilon_A, \ell\right)$-correctable and $B$ is $\left(\tilde\varepsilon_B, \ell\right)$-correctable, then $A \cup B$ is $\left(\tilde\varepsilon_A+\tilde\varepsilon_B, \ell\right)$-correctable.
\end{lem}
\noindent Letting $\ell=O(1)$, this lemma implies that for any contractible region of size $d$ linear in $n$, the error scaling is unchanged for different $d$. To see why $d$ can be relaxed to any linear function of $n$, suppose the system is $(\tilde\varepsilon,\ell)$-correctable for any contractible region of size $n/b$. A contractible region of size $n/c$ with $c<b$ can be obtained by conducting the expansion process for $b/c-1$ times, resulting in a region that is $(b\tilde\varepsilon/c,\ell)$-correctable. This indicates that the value of $c$ does not affect the error scaling, so  we can pick any $c>4$ and thus any $d$ linear in $n$.
\end{proof}

Note that our result does not hinge on a strict area law of the form of (\ref{eq:area}) and can allow the TEE to have small fluctuations depending on the subsystem. Specifically, suppose
\begin{equation}
S(A)=al(A)-b(A)\gamma(A)+ o(1),  
\end{equation}
then instead of (\ref{eq:TEEbound}) we obtain
\begin{equation}
    k\leq \sum_{i=1}^m(2\gamma(B_i)-\gamma(B_iC_i))-\gamma(A_1B_1).
\end{equation}
So  for sufficiently small fluctuations, we are still guaranteed a nontrivial TEE.

As mentioned in the main text, a corollary of this result that potentially has practical applications is that  trivial TEE signifies states that preclude good code properties.  A formal statement is as follows.
\begin{cor}
A 2D area-law state with zero TEE does not belong to any code that achieves $\varepsilon=o(1/n)$ on 
linear-size contractible regions.
\end{cor}

On the other hand, when the code error is large, the bound can become trivial in the sense that it does not provide a positive lower bound for TEE. We  now provide a physical demonstration of the phenomenon that the TEE vanishes upon increasing the code error, based on adding string tension in a string net model~\cite{Wen:ToricCodeTension}. Consider the following string net wave function for the toric code with string tension (defined on a $\sqrt{n}\times \sqrt{n}$ square lattice),
\begin{equation}
    |\Phi_{\alpha}^\mu\rangle=\sum_{c\in \alpha} e^{-\mu L(\phi_{c})}|\phi_{c}\rangle,
\end{equation}
where $L(\phi_{c})$ is the total string length of the configuration $\ket{\phi_c}$, $\mu$ is the string tension, and we are only summing over the equivalence class $\alpha$ of loops that are smoothly connected to each other. Let us consider the limit of infinite $\mu$, for which we can obtain analytical results. In this case, the string net state becomes either the all-zero state, or the state with the shortest non-contractible strings. We  regard the four $|\Phi_{\alpha}^\mu\rangle$'s as the basis for an $[n,2]$ code. For a $\sqrt{d}\times\sqrt{d}$ subsystem $A_d$, the reduced states of them are given by 
\begin{align}
    \Tr_{\ols{A_d}}\rho_0&=|0\rangle\langle 0|,\\
    \Tr_{\ols{A_d}}\rho_1&=\frac{1}{\sqrt{n}}\left(\sum_{x,y=1}^{\sqrt{d}}|x\rangle\langle y|+(\sqrt{n}-\sqrt{d})|0\rangle\langle0|\right),\\
    \Tr_{\ols{A_d}}\rho_2&=\frac{1}{\sqrt{n}}\left(\sum_{x,y=1}^{\sqrt{d}}|x^\perp\rangle\langle y^\perp|+(\sqrt{n}-\sqrt{d})|0\rangle\langle 0|\right),\\
    \Tr_{\ols{A_d}}\rho_3&=\frac{1}{N_C}\left(\sum_{C,C'}^{n_C}|C\rangle\langle C'|+(N_C-n_C)|0\rangle\langle0|\right),
\end{align}
where $\ket{0}$ denotes the state with no loop,    $\ket{x}$  ($\ket{x^\perp}$) denote the states with straight horizontal (vertical) lines at location $x$, and %$\ket{C}$ denotes the state with a shortest loop $C$ that is non-contractible both vertically and horizontally and intersects with $A_d$. 
{$\ket{C}$ denotes the state with a shortest loop $C$ that intersects with $A_d$ and is non-contractible both vertically and horizontally; that is, the loop $C$ is a zigzag line that encircles the torus  in both the horizontal and vertical directions.}
%\weicheng{I prefer sth like $x_H$, $x_V$ and $x_{HV}$ to make the notation consistent} \jinmin{We can change $x$ and $x^\perp$ to $x_H$, $x_V$ but leave $C$ as it is, cuz $C$ is not labelled by a coordinate} 
The number of such loops is denoted by $n_C$, and the total number of loops that are smoothly connected to them is denoted by $N_C$. 
Therefore, the reduced state of the reference state $\Gamma$ on $A_d$ is given by
\begin{align}
    \Tr_{\ols{A_d}}\Gamma&=\frac{1}{4}\Tr_{\ols{A_d}}(\rho_0+\rho_1+\rho_2+\rho_3)\\
    &=\left(1-\frac{\sqrt{d}}{2\sqrt{n}}-\frac{n_C}{N_C}\right)|0\rangle\langle0|+\frac{1}{4\sqrt{n}}\sum_{x,y=1}^{\sqrt{d}}(|x\rangle\langle y|+|x^\perp\rangle\langle y^\perp|)+\frac{1}{4N_C}\sum_{C,C'}^{n_C}|C\rangle\langle C'|.
\end{align}
By simple combinatorics, we find
\begin{align}
    N_C&=O(n^{1/4}2^{2\sqrt{n}}),\\
    n_C&=\Omega\left(\sqrt{\frac{d}{n}}N_C\right).
\end{align}
Since the number of $\ket{C}$'s for which there exists some $x$ with $\ket{C}=\ket{x}$ or $\ket{C}=\ket{x^\perp}$ is only $\mathrm{poly}(d)$, their contributions are negligible, so the subsystem variance is bounded as
\begin{align}
        \varepsilon_{{\mathsf{G}}_2}(\mathfrak{C},d)& \geq\onenorm{ \Tr_{\ols{A_d}}\rho_0-\Tr_{\ols{A_d}}\Gamma}\sim\frac{\sqrt{d}}{2\sqrt{n}}+\frac{n_C}{N_C}+\frac{1}{2\sqrt{n}}+\frac{1}{4N_C}=\Omega\left(\sqrt{\frac{d}{n}}\right), \label{eq:stringnet-variance}
\end{align}
which falls outside the nontrivial regime.  Indeed, the TEE of this model was shown to vanish~\cite{Wen:ToricCodeTension}, thus providing a physically interesting example of turning off TEE by increasing code error.  
However, observe that there is a gap between the properties of this extreme case and those required for the TEE bound. 
Further understanding of the situation of general $\mu$ may close the gap, which we leave for future work.

%\section{On code error from approximate KL}
%Start from approximate KLs (I suppose this is what we naturally get in physical scenarios like topo order, string net, CFT etc) of the form
%\begin{equation}
%    \bra{\phi_i}E\ket{\phi_j} = C_E\delta_{ij} + \epsilon_{ij},
%\end{equation}
%$\phi_{i,j}$ runs over a basis of the code subspace (size $2^k$), $E$ error operator with size of support $\leq d$, $\epsilon_{ii}$ are diagonal fluctuations, $\epsilon_{ij}$ off-diagonal errors, $C_E$ constant depending on $E$.

%\begin{equation}
%    \onenorm{\mathcal{B}} \leq \sum_{k,l}\onenorm{B_{ij}}\leq \sum_{ijkl}\epsilon_{ij,kl}
%\end{equation}
%where $kl$ label error Kraus operators so $E=E_k E_l$.
%$ij$ runs over logical subspace basis so $2^k$, if $kl$ runs over complete operator basis like Pauli then $4^d$.  

%Like in (\ref{eq:delta_psi_upper})
%\begin{equation}
%    \onenorm{\psi_R - \Gamma_R} \leq \onenorm{\mathcal{B}},
%\end{equation}
%and our complexity bounds work when $\onenorm{\psi_R - \Gamma_R}$ scales better than $n^{-1}$ (roughly).   So if we take something like $d = \log\log n$, and say one can argue that for size-$d$ operators all $\epsilon = O(n^{-\varDelta})$ then the $2^{O(d)}$ factors won't affect the exponent, so $\onenorm{\psi_R - \Gamma_R} = \tilde{O}(n^{-\varDelta})$, and we want this $<O(n^{-1})$ namely $\varDelta < -1$.   There is also a $2^{O(k)}$ factor but I suppose we can just make it $O(1)$.

%We can still take $d = \log^q n$ then $q$ will affect the exponent.  I haven't bothered carefully arguing the $k,d$ factors but the scaling arguments should roughly go like above.  

\section{CFT as AQEC code}\label{app:cft}

Here we provide details for the analysis of the code properties of CFT codes. % composed with low energy sector states. 
The gist of the analysis is that the subsystem variance of CFT codes can be estimated by calculating one-point functions that are closely related to the violation of the Knill--Laflamme conditions, using  methods native to the study of CFT.
A key conclusion is that the CFT codes are AQEC codes by nature, exhibiting  polynomial code error with the exponent  determined by the lowest scaling dimension of the CFT.

Specifically, we consider the low energy sector of some critical system described by a CFT. Let the system be defined on a hypersphere $S^D$ of $D$ spatial dimensions, and let the code subspace be composed of the first $2^k$ low lying states. Here the base manifold $S^D$ ensures the viability of the state-operator correspondence between a code state $|\phi_\alpha\rangle$ and a CFT operator $\phi_\alpha$ \cite{di1997conformal}, which means that on the manifold $S^D$, a low lying state $|\phi_\alpha\rangle$ is in one-to-one correspondence with the operator $\phi_\alpha$ inserted at the origin of the hyperplane $R^{D+1}$.

%\weicheng{The following sentence is not very necessary.}The correspondence works by considering the conformal transformation from the cylinder $S^D\times \mathbb{R}$ to the hyperplane $R^{D+1}$. $(\tau,\Omega)$.  
%in Here we can let $r$ be  with $r=e^{\tau/R}$. 
%A state on the cylinder is uniquely determined by the boundary condition at the left infinity end, which transforms to the origin of the hyperplane Thus the state is in one-to-one correspondence with the operator inserted at the origin of the hyperplane. 

We now proceed to the calculation of the scaling of subsystem variance. We use ``$\sim$'' to indicate approximate equivalence in terms of the scaling to the leading order or the next leading order, %\weicheng{Is G2 contribution of scaling to the leading order?}\jinmin{just "approximate equivalence in terms of the scaling" ?}, 
which will suffice in our scaling analysis. In order to estimate the subsystem variance on a region $R$ of $d$ qubits, we consider a complete set of operators $\mathcal{O}_i$ supported on the region $R$, such that $\Tr\left(\mathcal{O}_i \mathcal{O}_j\right) = \delta_{ij}$. Given a CFT code state $|\psi\rangle$ and the reference state $\Gamma$ defined by the equal mixture of all code states as in (\ref{eq:ref_state}), we have the expansion
\begin{align}
    \psi_R-\Gamma_R=\sum_i\Tr(\mathcal{O}_i(\psi_R-\Gamma_R))\mathcal{O}_i
    =\sum_i\bra{\psi}\mathcal{O}_i\ket{\psi}-\frac{1}{2^k}\sum_i\sum_{\phi\in\mathfrak{C}}\bra{\phi}\mathcal{O}_i\ket{\phi}\,.\label{eq:opexpansion}
\end{align}
Thus, we can study  subsystem variance by estimating the one-point functions $\langle\phi_\beta|\mathcal{O}|\phi_\gamma\rangle$ for size-$d$ operators $\mathcal{O}$. %Note that on $S^D\times \mathbb{R}$ this one-point function can be transformed into a correlation function on $\mathbb{R}^{D+1}$. 
For our purpose of complexity arguments, it suffices to choose $d\sim \mathrm{loglog}(n)$.  Such operators have sufficiently small sizes compared to the system size and are thus expected to be well approximated by a product of several local operators via proper renormalization group flow,  namely, $\mathcal{O}\sim \mathcal{O}_1(\Omega_1)\ldots\mathcal{O}_i(\Omega_i)$, where each $\mathcal{O}_i(\Omega_i)$ is a local operator supported at a finite region $\Omega_i$. 
Any local operator $\mathcal{O}_i(\Omega_i)$ can be decomposed as a linear combination of CFT scaling operators $\psi_\alpha$ as follows
\begin{equation}
    \mathcal{O}_i\sim\sum_{\alpha}a_{i,\alpha}\psi_\alpha,
\end{equation}
where $\psi_\alpha$ are the operators in the conformal tower, including the identity, the stress-energy tensors, the primary operators and their descendants. Therefore, $\mathcal{O}$ can be expanded in terms of products of CFT scaling operators:
\begin{equation}
    \mathcal{O}\sim\sum_{\alpha_1\ldots\alpha_i}a_{\alpha_1\ldots\alpha_i}\psi_{\alpha_1}(\Omega_1)\ldots\psi_{\alpha_i}(\Omega_i)
\end{equation}
It is clear that the term involving the identity operator only do not contribute in  \eqref{eq:opexpansion}, so we can subtract the identity component from $\mathcal{O}$ and focus on the correlation function of the rest $\mathcal{O}'$, i.e.,
\begin{equation}
    \mathcal{O}'\sim\sum_{\alpha'}a_{\alpha_1\ldots\alpha_i}\psi_{\alpha_1}(\Omega_1)\ldots\psi_{\alpha_i}(\Omega_i).
\end{equation}
where the summation $\sum_{\alpha'}$ excludes the component where all $\psi_{\alpha_i}$ are identities.

%%$\langle\phi_\beta|\mathcal{O}'|\phi_\gamma\rangle$
%the operator $\mc{O}'$ where  the identity component is subtracted from the $\mathcal{O}$ operator, 

Now we have \cite{di1997conformal, zou2020}
\begin{equation}
\langle\phi_\beta|\mathcal{O}'|\phi_\gamma\rangle\sim\sum_{\alpha'}\frac{a_{\alpha_1\ldots\alpha_i}}{\left(n^{1/D}\right)^{\Delta_{\alpha_1}+\ldots+\Delta_i}}\langle\phi_{\beta}(\infty)\psi_{\alpha_1}(\Omega_1)\ldots\psi_{\alpha_i}(\Omega_i)\phi_\gamma(0)\rangle,
\end{equation}
where $n^{1/D}$ is the length scale of the hypersphere $S^D$ in $D$ spatial dimensions and $\Delta_{\alpha_j}$ is the scaling dimension of $\psi_{\alpha_i}$. The one-point function becomes the largest when one of the $\psi_{\alpha_i}$'s is a primary field with the lowest scaling dimension $\varDelta$ while others are identities. %The one-point function becomes the largest when $\mathcal{O}'$ can be approximated by a primary field with the lowest scaling dimension $\varDelta$. 
That is,
\begin{equation}
    \langle\phi_\beta|\mathcal{O}'|\phi_\gamma\rangle\lesssim\frac{1}{n^{\varDelta/D}}\,.
\end{equation}
%where we used $L\sim n^{1/D}$  for $D$ spatial dimensions.
%, and the $n^{-\varDelta/D}$ scaling is attained when $\mathcal{O}'$ is a primary operator with the lowest scaling dimension.
Correspondingly, the scaling of the components of $\psi_R-\Gamma_R$ under the operator basis is bounded by
\begin{equation}
    \Tr(\mathcal{O}_i(\psi_R-\Gamma_R))\lesssim\frac{1}{n^{\varDelta/D}}.
\end{equation}
As a result, by (\ref{eq:opexpansion}) we obtain
\begin{equation}
\varepsilon \sim  \onenorm{\psi_R-\Gamma_R}\lesssim 2^{O(d)}\frac{1}{n^{\varDelta/D}},
\end{equation}
where the $2^{O(d)}$ factor is the number of terms in the summation, since the scaling of the summation is dominated by the worst term. In the case of $d\sim \mathrm{loglog}(n)$, this factor is only $\mathrm{polylog}(n)$ and thus insignificant in the overall scaling, ensuring that $\varepsilon=\Theta(n^{-\varDelta/D})$.
For the QEC inaccuracy, we obtain $\tilde\varepsilon\gtrsim{1}/{n^{\varDelta/D}}$ for local replacement noise channels by Proposition~\ref{thm:two-way}.
The superconstant all-to-all complexity is hence guaranteed under the condition $\varDelta > D$.  

Note that $\varDelta$ represents an inherent property of CFT and is related to the mass of bulk fields in AdS/CFT \cite{GUBSER1998105,Witten1998}.   As motivated in the main text, this situation is interesting especially from gravity perspectives.

\end{document}